\pdfoutput=1
\RequirePackage{ifpdf}
\ifpdf 
\documentclass[pdftex]{sigma}
\else
\documentclass{sigma}
\fi

\numberwithin{equation}{section}

{\theoremstyle{definition}
\newtheorem{Def}{Definition}[section]
\newtheorem{Rm}[Def]{Remark}}

\newtheorem{Prop}[Def]{Proposition}

\newtheorem{Thm}[Def]{Theorem}
\newtheorem{Cor}[Def]{Corollary}

\usepackage{euscript}
\usepackage{shuffle}
\usepackage{ulem}
\usepackage[all]{xy}
\newcommand{\interior}[1]{%
	{\kern0pt#1}^{\mathrm{o}}%
}

\begin{document}
\allowdisplaybreaks

\newcommand{\arXivNumber}{1809.00534}

\renewcommand{\PaperNumber}{108}

\FirstPageHeading

\ShortArticleName{Controlled Loewner--Kufarev Equation Embedded into the Universal Grassmannian}

\ArticleName{Controlled Loewner--Kufarev Equation Embedded\\ into the Universal Grassmannian}

\Author{Takafumi AMABA~$^\dag$ and Roland FRIEDRICH~$^\ddag$}

\AuthorNameForHeading{T.~Amaba and R.~Friedrich}

\Address{$^\dag$~Fukuoka University, 8-19-1 Nanakuma, J\^onan-ku, Fukuoka, 814-0180, Japan}
\EmailD{\href{mailto:fmamaba@fukuoka-u.ac.jp}{fmamaba@fukuoka-u.ac.jp}}

\Address{$^\ddag$~ETH Z\"urich, D-GESS, CH-8092 Zurich, Switzerland}
\EmailD{\href{mailto:roland.friedrich@gess.ethz.ch}{roland.friedrich@gess.ethz.ch}}

\ArticleDates{Received June 30, 2020, in final form October 22, 2020; Published online October 29, 2020}

\Abstract{We introduce the class of controlled Loewner--Kufarev equations and consider aspects of their algebraic nature. We lift the solution of such a controlled equation to the (Sato)--Segal--Wilson Grassmannian, and discuss its relation with the tau-function. We briefly highlight relations of the Grunsky matrix with integrable systems and conformal field theory. Our main result is the explicit formula which expresses the solution of the controlled equation in terms of the signature of the driving function through the action of words in generators of the Witt algebra.}


\Classification{35Q99; 30F10; 35C10; 58J65}

\section{Introduction}

C.~Loewner~\cite{Lo} and P.P.~Kufarev~\cite{Ku}
initiated a theory which was then further extended
by C.~Pommerenke~\cite{Po2}, and which shows that
given any continuously increasing family of
simply connected domains containing the origin in
the complex plane,
the inverses of the Riemann mappings associated to the domains
are described by a partial differential equation,
the so-called
{\it Loewner--{\rm (}Kufarev{\rm )} equation}
\begin{equation*}
\frac{\partial}{\partial t}
f_{t}(z)
=
z
f_{t}^{\prime} (z)
p(t,z),
\end{equation*}
where the $f_t$ are the inverses of the Riemann map and $p(z,t)$ is a function with positive real part (see Section~\ref{LK-asCLK} for details).
More recently,
I.~Markina
and
A.~Vasil'ev~\cite{MaVa10,MaVa16}
considered the so-called {\it alternate Loewner--Kufarev equation},
which describes not necessarily increasing chains of domains.

We introduce a further generalisation, namely, the class of {\it controlled Loewner--Kufarev equations}
\begin{equation*}
\mathrm{d} f_{t}(z)
=
z f_{t}^{\prime} (z)
\{
	\mathrm{d} x_{0} (t)
	+
	\mathrm{d} \xi ( \mathbf{x}, z )_{t}
\},
\qquad
f_{0} (z) \equiv z \in \mathbb{D},
\end{equation*}
where
$\mathbb{D}$
is the unit disc in the complex plane centred at zero,
$x_{0}, x_{1}, x_{2}, \ldots $
are given functions which will be called the
{\it driving functions},
$
\mathbf{x}
=
( x_{1}, x_{2}, \ldots )
$
and
$
\xi ( \mathbf{x}, z )_{t}
:=
\sum_{n=1}^{\infty} x_{n}(t) z^{n}
$.
The controlled Loewner--Kufarev equation can be transformed, after a calculation, into
\begin{equation*}
\mathrm{d} f_{t}(z)=-\sum_{n=0}^{\infty}
( L_{n}f ) (z) \mathrm{d} x_{n},
\end{equation*}
where the $L_{n} := - z^{n+1} \partial / ( \partial z )$,
$n \in \mathbb{Z}$, are the generators of
the Witt algebra, i.e., the
central charge zero Virasoro algebra, satisfying the commutation relations
\begin{equation}\label{Witt_alg}
[L_m,L_n]=(m-n)L_{m+n}.
\end{equation}
Therefore, we are going to consider an extension of~\cite{Fr10}, where the second author established and studied the role of Lie vector fields, boundary variations and the Witt algebra in connection with the Loewner--Kufarev equation.

Let us recall first some of the classical work of
A.A.~Kirillov
and
D.V.~Yuriev
\cite{KiYu}
/
G.B.~Segal
and
G.~Wilson
\cite{SW}
/
N.~Kawamoto, Y.~Namikawa, A.~Tsuchiya and Y.~Yamada~\cite{KNTY} which will be also fundamental in the present context, in particular in understanding the appearance of the Virasoro algebra with nontrivial central charge.

A.A.~Kirillov and D.V.~Yuriev~\cite{KiYu},
constructed a highest weight representation
of the Virasoro algebra,
where the representation space is given by
the space of all holomorphic sections of an analytic line bundle
over the orientation-preserving diffeomorphism group
$\mathrm{Diff}_{+} S^{1}$
of the unit circle~$S^{1}$
(modulo rotations).
They also gave an embedding of
$\big(\mathrm{Diff}_{+} S^{1} \big)/ S^{1}$
into the infinite dimensional Grassmannian.
In fact, this embedding is an example
of a construction of solutions to the KdV hierarchy
found by I.~Krichever~\cite{Kr3},
which we address in Section~\ref{Krichever_Constr}.
If we embed a~univalent function on the unit disc~$\mathbb{D}$
into the infinite dimensional Grassmannian,
by the methods of
Kirillov--Yuriev~\cite{KiYu},
Krichever~\cite{Kr3},
or
Segal--Wilson~\cite{SW},
then one needs to track the
Faber polynomials and Grunsky coefficients
associated to the univalent function.
In general, it is not easy to calculate them from the definition.
One of our main results is, however, the following.

\begin{Thm}[see Propositions~\ref{Faber-variation} and~\ref{explicit-Grunsky}]
The Faber polynomials and Grunsky coefficients
associated to solutions of the controlled Loewner--Kufarev equation
satisfy linear differential equations, and
the Grunsky coefficients can be explicitly calculated.
\end{Thm}

In~\cite{Fr10}, the second author proposed to lift
the embedded Loewner--Kufarev equation to the
determinant line bundle over the Sato--Segal--Wilson Grassmannian
$\mathrm{Gr}(H)$,
as a natural extension of the ``Virasoro uniformisation"
approach by M.~Kontsevich~\cite{Ko} / R.~Friedrich and J.~Kalkkinen~\cite{FrKa}
to construct generalised stochastic / Schramm--Loewner evolutions~\cite{Sch}
on arbitrary Riemann surfaces, which would also yield
a connection with conformal field theory
in the spirit of \cite{KNTY, SW}.
Let us also mention the work of B.~Doyon \cite{Do1}, who uses conformal loop ensembles~(CLE), and which is related to the content of the present article.

In \cite{MaVa16},
I.~Markina
and
A.~Vasil'ev established basic parts of this program, by considering
embedded solutions to the Loewner--Kufarev equation
into the Segal--Wilson Grassmannian
and related the dynamics therein with the representation
of the Virasoro algebra, as discussed by
Kirillov--Yuriev \cite{KiYu}.
Further, they considered the tau-function associated to the embedded solution
as a lift to the determinant line bundle.
As observed and briefly discussed in~\cite{FrKa, Ko}, the generator of the stochastic Loewner equation is {\it hypo-elliptic}.

I.~Markina, I. Prokhorov and A.~Vasil'ev~\cite{MaDmVa} observed and discussed the sub-Riemannian nature of the coefficients of univalent functions.
As the second author pointed out~\cite{Fr10}, this connects with the general theory of hypo-elliptic flows, as explained in the book by~F. Baudoin~\cite{Ba}, and led him to propose a connection of the (stochastic) Loewner--Kufarev equation with rough paths.
Now, in
the theory of rough paths
(see, e.g., the introduction in \cite{LyCaLh}),
one of the central objects of consideration is
the following controlled differential equation:
\begin{equation}\label{CDE}
\mathrm{d} Y_{t} = \varphi (Y_{t}) \mathrm{d} X_{t},
\end{equation}
where $X_{t}$ is a continuous path in a normed space $V$,
called the {\it input} of (\ref{CDE}).
On the other hand,
the path $Y_{t}$ is called the {\it output} of (\ref{CDE}).
When we deal with this equation,
an important object is
the {\it signature} of the input $X_{t}$,
with values in the (extended) tensor algebra associated with $V$
and which is written in the following form:
\begin{equation*}
S(X)_{s,t}:=\big(1, X_{s,t}^{1}, X_{s,t}^{2}, \ldots , X_{s,t}^{n}, \ldots \big),\qquad s \leqslant t.
\end{equation*}
If $X_{t}$ has finite variation with respect to $t$,
then each $X_{s,t}^{n}$ is the $n$th iterated integral of
$X_{t}$ over the interval~$[s,t]$.
With this object,
a combination of the
{\it Magnus expansion}
and the
{\it Chen--Strichartz expansion theorem}
(see, e.g., \cite[Section~1.3]{Ba})
tells us that
the output $Y_{t}$ is given as the result of the action of
$S(X)_{0,t}$
applied to $Y_{0}$.
Heuristically, we may say that
a `group element'~$S(X)$ in some big `group'
acts on some element in the
(extended)
tensor algebra
$T(\hspace{-1mm}(V)\hspace{-1mm})$
which gives the output $Y_{t}$,
or it might be better to say that the vector field $\varphi$ defines
how the `group element' acts on the algebra.
In this spirit,
we would like to describe such a picture
in the context of controlled Loewner--Kufarev equations.

For this, we extract the algebraic structure
of the controlled Loewner--Kufarev equation.
If we regard the driving functions
$x_{0}, x_{1}, x_{2}, \ldots$
just as letters in an alphabet
then it turns out that
explicit expressions for
the associated Grunsky coefficients are given by
the algebra of formal power series,
where the space of coefficients is given by words over this alphabet.
It is worth mentioning that
the action of the words over this alphabet will be actually given by
the {\it negative} part of the Witt generators.
Thus the action of the signature encodes
many actions of such negative generators.
This can be used to derive a formula for
$f_{t}(z)$
as the signature `applied' to the initial data $f_{0} (z) \equiv z$
(see Theorem~\ref{Sol-by-Witt}).

Now, given a diffeomorphism of the unit circle $S^{1}$,
the solution to the associated conformal welding problem
is a solution to the dispersionless Toda lattice hierarchy~\cite{TT91, WZ00}.
Also in this case, the corresponding tau-function
is described by the (full) Grunsky coefficients
and this generates the solution via an explicit formula.
This gives us the possibility to explicitly describe
the solution to the conformal welding problem
associated to Malliavin's canonic diffusion~\cite{Ma99}
by means of a controlled Loewner--Kufarev equation; a topic to
which we intend to return elsewhere.
Since the canonic diffusion is a natural object `on' the diffeomorphism group of~$S^{1}$,
as well as the Brownian motion on a Euclidean space,
it would describe a natural universal class in the infinite-dimensional situation.

However, the story so far lets us ask
how the signature associated with the driving functions
describes the corresponding tau-function
rather than~$f_{t}$ itself.

\begin{Thm}[see Theorem~\ref{Tau-by-Witt}]
Along the solution of the controlled Loewner--Kufarev equation, the associated tau-function can be written as the determinant of a quadratic form of the signature.
\end{Thm}

Let us now summarise the structure of the paper.
In Section~\ref{Sec_CLK},
we formulate solutions~$f_{t}(z)$
to controlled Loewner--Kufarev equations.
We add also a brief review of the classical
Loewner--Kufarev equation,
and then explain how the classical one is recovered
from the controlled Loewner--Kufarev equation.
We track the variation of the Taylor-coefficients of~$f_{t}$
and also the Faber polynomials and Grunsky coefficients.
In Section~\ref{Sec_LK/Gr},
we first recall briefly basics of the Segal--Wilson Grassmannian
and Krichever's construction.
After that, we describe how a univalent function
on~$\mathbb{D}$
is embedded into the Grassmannian.
We extract the algebraic structure of
the controlled Loewner--Kufarev equation in order to obtain
Theorem~\ref{Tau-by-Witt}. In
Appendix~\ref{Appdx}, we give the proofs of Theorems~\ref{explicit-Taylor} and~\ref{Tau-by-Witt}, respectively, and of Proposition~\ref{explicit-Grunsky}.

\section{The controlled Loewner--Kufarev equation}\label{Sec_CLK}

General assumption: $\mathbb{N}$ denotes the set of all positive integers, i.e., $\{1, 2, 3, \ldots\}$, (without zero).

\subsection{Definition of solutions to controlled Loewner--Kufarev equations}\label{Def-LK}

Given functions
$x_{1}, x_{2}, \ldots \colon [0,T] \to \mathbb{C}$,
we will write
$
\mathbf{x}
:=
( x_{1}, x_{2}, \ldots )
$
and
\begin{equation*}
\xi ( \mathbf{x}, z )_{t}
:=
\sum_{n=1}^{\infty}
x_{n} (t) z^{n},
\qquad
\text{for $z \in \mathbb{C}$},
\end{equation*}
if it converges.
If $A\colon [0,T] \to \mathbb{C}$ is of bounded variation,
we write $\mathrm{d} A$ or $A( \mathrm{d} t )$
(when emphasising the coordinate $t$ on $[0,T]$) for
the associated
complex-valued Lebesgue--Stieltjes measure
on $[0,T]$,
and the total variation measure will be denoted by
$\vert \mathrm{d} A \vert$.

\begin{Def}\label{Def:solution} Let $T>0$.
Suppose that
$x_{0}\colon [0,T] \to \mathbb{R}$, as well as
$x_{1}, x_{2}, \ldots\colon [0,T] \to \mathbb{C}$,
are continuous and of bounded variation,
and $x_{0} (0) = 0$.
Let
$f_{t}\colon \mathbb{D} \to \mathbb{C}$
be conformal mappings for $0 \leqslant t \leqslant T$.
We say $\{ f_{t} \}_{0 \leqslant t \leqslant T}$ is
a {\it solution} to
\begin{equation}\label{controlled-LK}
\mathrm{d} f_{t}(z)
=
z f_{t}^{\prime} (z)
\{	\mathrm{d} x_{0} (t)
	+
	\mathrm{d} \xi ( \mathbf{x}, z )_{t}
\}, \qquad
f_{0}(z) \equiv z \in \mathbb{D}
\end{equation}
if
\begin{itemize}\itemsep=0pt
\item[(1)]
$f_{0} (z) \equiv z$ for $z \in \mathbb{D}$,

\item[(2)]
$
\sum_{n=1}^{\infty}
n
\int_{[0,T]}
\vert
	\mathrm{d} x_{n}
\vert
(t)
r^{n}
$
converges for all $r \in (0,1)$,

\item[(3)]
for each compact set $K \subset \mathbb{D}$,
the mapping
$
[0,T] \ni t \mapsto f_{t}^{\prime}\vert_{K} \in C(K)
$
is continuous
with respect to the uniform norm on $K$,

\item[(4)]
it holds that
\begin{equation*}
f_{t} (z) - z
=
\int_{0}^{t}
z f_{s}^{\prime} (z)
\big\{
	\mathrm{d} x_{0} (s)
	+
	\mathrm{d} \xi ( \mathbf{x}, z )_{s}
\big\},
\qquad
(t,z) \in [0,T] \times \mathbb{D}.
\end{equation*}
\end{itemize}
\end{Def}
In the sequel, we refer to equation (\ref{controlled-LK})
as a {\it controlled Loewner--Kufarev equation}
(with driving paths $x_{0}$ and $\mathbf{x} := (x_{1},x_{2},\ldots)$).

In joint work with T.~Murayama~\cite{AFM}, we proved that
a solution to the controlled Loew\-ner--Ku\-fa\-rev equation
is unique if it exists \cite[Theorem~3.1]{AFM}.
In the $\omega$-controlled case,
for \mbox{$\omega (0,T) < 1/2$},
a solution exists
(and hence uniquely exists), cf.~\cite[Theorem~3.2]{AFM}. More specifically, we have

\begin{Prop}[{\cite[Lemma~2.1]{AFM}}]\label{Prop:Maruyama}Under the assumptions {\rm (1)--(3)} above,
\begin{itemize}\itemsep=0pt
\item[$(i)$]
the series
$\xi (\mathbf{x}, z )_{t}$
in $z$ has convergence radius one for each $t \in [0,T]$,

\item[$(ii)$]
the family
$
\{ \xi (\mathbf{x}, z) \}_{0 \leqslant t \leqslant T}
$
of holomorphic functions on $\mathbb{D}$
is continuous in the topology of locally uniform convergence,

\item[$(iii)$]
the function
$
t \mapsto \xi ( \mathbf{x}, z )_{t}
$
is of bounded variation and satisfies
\begin{equation*}
\mathrm{d} \xi ( \mathbf{x}, z )_{t}
=
\sum_{k=1}^{\infty}
z^{k} \mathrm{d} x_{k} (t),
\end{equation*}
for each $z \in \mathbb{D}$.
\end{itemize}
\end{Prop}
Furthermore, in \cite[equations (3.1) and (3.2)]{AFM} we proved that $f_{t}^{\prime}(0) = \mathrm{e}^{x_{0}(t)} > 0.$

\begin{Def}\label{Def:univalent_sol}We say $\{ f_{t} \}_{0 \leqslant t \leqslant T}$ is a~{\it univalent solution}
to the controlled Loewner--Kufarev equation
if it is a solution to~(\ref{controlled-LK})
and $f_{t}$ is a univalent function on~$\mathbb{D}$
for each $0 \leqslant t \leqslant T$.
\end{Def}

\subsection{Loewner--Kufarev equation as a controlled Loewner--Kufarev equation}\label{LK-asCLK}

\begin{Def} Suppose that $\Omega (t) \subset \mathbb{C}$ is given
for each $0 \leqslant t \leqslant T$.
$\{ \Omega (t) \}_{0 \leqslant t \leqslant T}$
is called a~{\it Loewner subordination chain}
if
\begin{itemize}\itemsep=0pt
\item[(1)]
$0 \in \Omega (s) \subsetneq \Omega (t)$
for each $0 \leqslant s < t \leqslant T$,

\item[(2)]
$\Omega (t)$ is a simply connected domain
(i.e., open, connected and simply connected)
for each $t \in [0,T]$,

\item[(3)]
(Continuity in the sense of Carath\'eodory,
under the conditions (1) and (2)):
For each $t \in [0,T]$
and any sequence $0 \leqslant t_{n} \uparrow t$,
$\cup_{n=1}^{\infty} \Omega (t_{n}) = \Omega (t)$.
\end{itemize}
\end{Def}

For the following Definition~\ref{LC}, cf.\ specifically
{\cite[Chapter~6, Section~6.1, pp.~156--157; Chapter~2, Section~2.1, p.~35 and Lemma~2.1]{Po2}}.
\begin{Def}[\cite{Po2}]\label{LC}Let
$f_{t}\colon \mathbb{D} \to \mathbb{C}$
be given for $0 \leqslant t \leqslant T$.
Then
$\{ f_{t} \}_{0 \leqslant t \leqslant T}$
is called a~{\it Loewner chain}
if
\begin{itemize}\itemsep=0pt
\item[(1)]
$f_{t}$ is analytic and univalent on $\mathbb{D}$,
for each $0 \leqslant t \leqslant T$,

\item[(2)]$f_{t} (z)=\mathrm{e}^{t} z + a_{2} (t) z^{2} + \cdots$, for
$z \in \mathbb{D}$,

\item[(3)] $f_{s} (\mathbb{D}) \subset f_{t} (\mathbb{D})$,
for each $0 \leqslant s < t \leqslant T$.
\end{itemize}
\end{Def}

The above chains
$\{ \Omega (t) \}$
and
$\{ f_{t} \}$
are known to be in one-to-one correspondence
via the relation
$
\Omega ( \tau )
=
f_{t} ( \mathbb{D} ),
$
where $t = \log f_{\tau}^{\prime} (0)$
is a time-reparametrisation to satisfy
Definition~\ref{LC}(2)
(see \cite[Chapter~6, Section~6.1]{Po2}).

\begin{Thm}[{\cite[Theorem~6.2]{Po2}}]
Let $f_{t}\colon \mathbb{D} \to \mathbb{C}$
be given for $0 \leqslant t \leqslant T$.
Then
$\{ f_{t} \}_{0 \leqslant t \leqslant T}$
is a~Loewner chain if and only if
there exist
constants
$r_{0}, K_{0} > 0$,
and
a function $p(t,z)$,
analytic in $z \in \mathbb{D}$,
and measurable in $t \in [0,T]$
such that
\begin{itemize}\itemsep=0pt
\item[$(i)$]
for each $0 \leqslant t \leqslant T$,
the function
$
f_{t} (z)
=
\mathrm{e}^{t} z + \cdots
$
is analytic in $\vert z \vert < r_{0}$,
the mapping
$[0,T] \ni t \mapsto f_{t} (z)$
is absolutely continuous for each $\vert z \vert < r_{0}$,
and
\begin{equation*}
\vert f_{t} (z) \vert
\leqslant
K_{0} \mathrm{e}^{t},
\qquad
\text{for all
$\vert z \vert < r_{0}$
and
$t \in [0,T]$.}
\end{equation*}

\item[$(ii)$]
$\mathrm{Re} \{ p(t,z) \} > 0$,
for all
$(t,z) \in [0,T] \times \mathbb{D}$,
and
\begin{equation}\label{PomLK}
\frac{\partial}{\partial t}
f_{t}(z)
=
z
f_{t}^{\prime} (z)
p(t,z),
\end{equation}
for all $\vert z \vert < r_{0}$ and for almost all $t \in [0,T]$.
\end{itemize}
\end{Thm}
According to the terminology in \cite{BCDV} we call the equation (\ref{PomLK}) the
{\it Loewner--Kufarev equation}
(if we regard $p(t,z)$ as given and $f_{t}(z)$ as unknown).

Because of equation (\ref{PomLK}), it holds that $p(t,0)=\lim\limits_{z \to 0} \big( \frac{\partial}{\partial t} f_{t}(z) \big)
/( z f_{t}^{\prime} (z) )= 1$,
and hence the `Herglotz representation theorem'
applies, which permits us to conclude that,
for every $t \in [0,T]$,
there exists a probability measure $\nu_{t}$
on $S^{1} = \partial \mathbb{D}$
(which is naturally identified with
$[0, 2 \pi]$ as measurable spaces,
and then the induced probability measure is still
denoted by $\nu_{t}$)
such that
\begin{equation*}
p(t,z)
=
\int_{0}^{2\pi} \hspace{-1mm}
\frac{ \mathrm{e}^{i \theta} + z }{ \mathrm{e}^{i \theta} - z }
\nu_{t} ( \mathrm{d} \theta )
\qquad
\text{for $z \in \mathbb{D}$.}
\end{equation*}
Substituting this into (\ref{PomLK}),
the Loewner--Kufarev equation becomes
\begin{equation}\label{LK} \begin{split}
\frac{ \partial f_{t} }{ \partial t } (z)
=z f_{t}^{\prime} (z)
\int_{0}^{2\pi} \hspace{-1mm}
\frac{
	\mathrm{e}^{i \theta} + z
}{
	\mathrm{e}^{i \theta} - z
}
\nu_{t} ( \mathrm{d} \theta ) .
\end{split}
\end{equation}
Assuming that
$
\nu_{t}( \mathrm{d} \theta )
=:
\nu_{t}( \theta ) \mathrm{d} \theta
$,
we write the Fourier series of $\nu_{t} (\theta )$ as
\begin{equation*}
\nu_{t} ( \theta )
=
\frac{1}{2\pi}
\left\{
a_{0}(t)
+
\sum_{k=1}^{\infty}
\big(
	a_{k}(t) \cos ( k \theta )
	+
	b_{k}(t) \sin ( k \theta )
\big)
\right\} .
\end{equation*}
We temporarily introduce the notation
$x_{0} (t) := \int_{0}^{t} a_{0} (s) \mathrm{d} s$
and
\begin{equation*}
u_{k} (t)
:=
\int_{0}^{t}
a_{k}(s)
\mathrm{d}s,
\qquad
v_{k} (t)
:=
- \int_{0}^{t}
b_{k}(s)
\mathrm{d} s,
\end{equation*}
for $k=1,2,\ldots$.
Because of the relations
\begin{equation*}
\frac{1}{2\pi}
\int_{0}^{2\pi}
\frac{
	\mathrm{e}^{ i \theta } + z
}{
	\mathrm{e}^{ i \theta } - z
}
\cos ( k \theta )
\mathrm{d} \theta
=
z^{k},
\qquad
\frac{1}{2\pi}
\int_{0}^{2\pi}
\frac{ \mathrm{e}^{ i \theta } + z }{ \mathrm{e}^{ i \theta } - z }
\sin ( k \theta )
\mathrm{d} \theta
=
- i z^{k},
\end{equation*}
for $k=1,2,\ldots$ and $\vert z \vert < 1$,
equation (\ref{LK}) assumes the following form:
\begin{equation*}
\frac{ \partial f_{t} }{ \partial t } (z)
=
z f_{t}^{\prime} (z)
\left\{
	\dot{x}_{0} (t)
	+
	\sum_{k=1}^{\infty}
	\big(
		\dot{u}_{k} (t) + i \dot{v}_{k} (t)
	\big)
	z^{k}
\right\} .
\end{equation*}
This can be rewritten as the following
controlled differential equation
\begin{equation*}
\mathrm{d} f_{t}(z)
=
z f_{t}^{\prime} (z)
\{\mathrm{d} x_{0} (t)
	+
	\mathrm{d} \xi ( \mathbf{x}, z )_{t}\},
\end{equation*}
where
$x_{k} (t) = u_{k} (t) + i v_{k} (t)$
for $k \geqslant 1$,
and
$
\xi ( \mathbf{x}, z )_{t}
=
\sum_{k=1}^{\infty} x_{k} (t) z^{k}
$.

If we omit the condition $\mathrm{Re} \{ p(t,z) \} > 0$,
that is, we allow the real part of $p(t,z)$ to have an arbitrary sign,
then equation (\ref{PomLK}) is called the
{\it alternate Loewner--Kufarev equation},
as considered by
I.~Markina and A.~Vasil'ev \cite{MaVa10}.
Intuitively, this
describes evolutions of conformal mappings
whose images of $\mathbb{D}$ are not necessary increasing, i.e., not strict subordinations. It appears that the general theory with respect to the existence and uniqueness of solutions
is not yet fully developed. However, our
controlled Loewner--Kufarev equation (\ref{controlled-LK})
deals with this alternate case because we have not assumed
that
$
p(t,z)
:=
\frac{\mathrm{d}}{\mathrm{d}t}
( x_{0} (t) + \xi ( \mathbf{x}, z )_{t} )
$
has a positive real part.

\begin{Rm}Readers focusing on radial Loewner equations
might feel puzzled by the heuristic assumption that
the Radon--Nikodym density
$
\frac{
	\nu_{t} ( \mathrm{d} \theta )
}{
	\mathrm{d} \theta
}
= \nu_{t} ( \theta )
$
exists,
because the radial Loewner equation describes the case
$
\nu_{t} ( \mathrm{d} \theta )
=
\delta_{ \mathrm{e}^{ i w(t) } } ( \mathrm{d} \theta )
$
where $w(t)$ is a continuous path in $\mathbb{R}$,
so that there does not exist a Radon--Nikodym density.
However, several explicit examples of Loewner--Kufarev equations
within this setting,
are presented with simulations in Sola~\cite{So}.
\end{Rm}

\subsection{Taylor coefficients along the controlled Loewner--Kufarev equation}

Suppose that
$x_{0}\colon [0, +\infty ) \to \mathbb{R}$,
$x_{1}, x_{2}, \ldots \colon [0, +\infty ) \to \mathbb{C}$
are continuous and of bounded variation.
Let $\{ f_{t} \}_{0 \leqslant t \leqslant T}$
be a solution to the controlled Loewner--Kufarev equation~(\ref{controlled-LK}).
We parametrise $f_{t}$ as
\begin{equation}\label{parametrization}
f_{t}(z)
=
C(t)
\big(
	z + c_{1}(t) z^{2} + c_{2}(t) z^{3} + c_{3}(t) z^{4} + \cdots
\big) ,
\end{equation}
with the additional convention that $c_{0}(t) \equiv 1$.

The dynamics of the coefficients
$(c_{1}, c_{2}, \ldots)$ has been previously
studied by Vasil'ev and his co-authors
\cite{HiMaVa, MaDmVa, MaVa10, MaVa11}.
The
(stochastic/Schramm)-Loewner
(equation/evolution)
(SLE)
case is discussed by
Friedrich~\cite{Fr10}. A complementary, conformal field theoretic perspective of the Bieberbach--de Branges theorem is given by Duplantier et al.~\cite{DNNZ}.
Within our framework, we get the following similarly:

\begin{Prop}\label{CLK-Hiera}
Let $\{ f_{t} \}_{0 \leqslant t \leqslant T}$ be a solution
to the controlled Loewner--Kufarev equation
\eqref{controlled-LK}
with the parametrisation~\eqref{parametrization}.
Then we have
\begin{equation*}
\mathrm{d} C(t)
=
C(t) \mathrm{d} x_{0} (t),
\end{equation*}
and
\begin{equation}\label{c_n}
\begin{cases}
\mathrm{d} c_{1} (t)
=
\mathrm{d} x_{1} (t)
+
c_{1} (t) \mathrm{d} x_{0} (t), \\
\mathrm{d} c_{2} (t)
=
\mathrm{d} x_{2} (t)
+ 2 c_{1} (t) \mathrm{d} x_{1} (t)
+ 2 c_{2} (t) \mathrm{d} x_{0} (t), \\
\hspace{30mm}\vdots \\
\displaystyle
\mathrm{d} c_{n} (t)
=
\mathrm{d} x_{n} (t)
+
\sum_{k=1}^{n-1}
(k+1) c_{k}(t)
\mathrm{d} x_{n-k} (t)
+
n c_{n}(t)
\mathrm{d} x_{0} (t), \qquad
\text{for $n \geqslant 2$,}
\end{cases}
\end{equation}
with the initial conditions
$C(0) = 1$ and
$c_{1} (0) = c_{2} (0) = \cdots = 0$.
In particular, $C= \{ C(t) \}_{0 \leqslant t \leqslant T}$
takes its values in $\mathbb{R}$.
\end{Prop} 
As
$
f_{t}^{\prime} (0)
=
C(t)
=
\mathrm{e}^{ x_{0} (t) - x_{0} (0) }
\neq 0
$,
we get
\begin{Cor}\label{loc-univ}Let $\{ f_{t} \}_{0 \leqslant t \leqslant T}$ be a~solution to the controlled Loewner--Kufarev equation~\eqref{controlled-LK}. Then $f_{t}$ is univalent in a neighbourhood of $0$,
for each $0 \leqslant t \leqslant T$.
\end{Cor}

\begin{Thm}\label{explicit-Taylor} Let $\{ f_{t} \}_{0 \leqslant t \leqslant T}$ be a solution
to the controlled Loewner--Kufarev equation
\eqref{controlled-LK}.
Then for
each $n \in \mathbb{N}$,
the coefficient $c_{n}$ in
\eqref{parametrization}
is given by
\begin{gather*}
c_{n} (t)=
\sum_{p=1}^{n}
\sum_{
	\substack{
		i_{1}, \ldots , i_{p} \in \mathbb{N}: \\
		i_{1} + \cdots + i_{p} = n
	}
}
\widetilde{w}(n)_{i_{1}, \ldots , i_{p}}
\mathrm{e}^{ n x_{0} (t) } \\
\hphantom{c_{n} (t)=}{}\times
\int\displaylimits_{0\leqslant s_{1} < s_{2} < \cdots < s_{p} \leqslant t}
\!\!\!\!\!\!\!\!\!\mathrm{e}^{ -i_{1} x_{0} (s_{1}) }
\mathrm{d} x_{i_{1}} (s_{1})
\mathrm{e}^{ -i_{2} x_{0} (s_{2}) }
\mathrm{d} x_{i_{2}} (s_{2})
\cdots
\mathrm{e}^{ -i_{p} x_{0} (s_{p}) }
\mathrm{d} x_{i_{p}} (s_{p}),
\end{gather*}
where
\begin{equation*}
\widetilde{w}(n)_{i_{1}, \ldots , i_{p}}
:=
\big\{ (n-i_{1}) + 1 \big\}
\big\{ (n-(i_{1}+i_{2})) + 1 \big\}
\cdots
\big\{ (n-(i_{1}+i_{2}+ \cdots +i_{p-1})) + 1 \big\} ,
\end{equation*}
and $n = i_{1} + \cdots + i_{p}$.
\end{Thm}
The proof can be found in Appendix~\ref{app:explicit-Taylor}.

\subsection[Variation of Grunsky coefficients induced by a Loewner--Kufarev equation]{Variation of Grunsky coefficients induced\\ by a Loewner--Kufarev equation}\label{Grunsky}

There are several different ways to introduce the {\it Faber polynomials}. Here we give a derivation by utilising Teo~\cite{Teo03}, and, an alternative one, in Section~\ref{FabGrunsky}, which serves our purpose better. For a (formal) power series $f(z)=a_1z+a_2z^2+a_3z^3+\cdots$, $a_1\neq0$,
the (generalised) {\it Faber polynomials} $Q_n(w)$, $n\in\mathbb{N}$, associated to $f$, are defined as
\begin{equation}\label{gen_Faber_poly}
\log\frac{w-f(z)}{w}=\log\frac{f(z)}{a_1 z}-\sum_{n=1}^{\infty}\frac{Q_n(w)}{n}z^n.
\end{equation}
By differentiating equation~(\ref{gen_Faber_poly}), and reordering it, we obtain, via the Residue theorem, the Faber polynomials (cf. also expression~(\ref{Faber_alt})), as
\begin{equation*}
Q_n(w)
=
\operatorname{Res}\displaylimits_{z=0}
\left[\frac{wz^{-n}}{w-f(z)}\frac{f'(z)}{f(z)}\right]{\rm d}z~
\underset{\zeta=f(z)}{=}~
\operatorname{Res}\displaylimits_{
\zeta = 0
}\left[\frac{\left(f^{-1}(\zeta)\right)^{-n}}{1-\zeta w^{-1}}\frac{1}{\zeta}\right]{\rm d}\zeta.
\end{equation*}
The coefficients $( b_{-m,-n} )_{m,n=1}^{\infty}$
in the series expansion
\begin{equation}
\label{Grunsky_coefficients}
\log \frac{ f(z) - f(\zeta ) }{ z - \zeta }
=-\sum_{m=0}^{\infty}\sum_{n=0}^{\infty}b_{-m,-n}z^{m}\zeta^{n},
\end{equation}
at $(z,\zeta ) = (0,0)$,
are called the
(generalised)
{\it Grunsky coefficients}
of $f$.
Equivalently, these are defined via the Laurent series
at $z=0$,
\begin{equation*}
Q_{n} ( f(z) )=z^{-n}+n\sum_{m=1}^{\infty} b_{-n,-m} z^{m}.
\end{equation*}
\begin{Prop}\label{inverse-evol}
Let $\{ f_{t} \}_{0 {\leqslant} t {\leqslant} T}\!$ be a solution to the controlled Loewner--Kufarev equa\-tion~\eqref{controlled-LK}.
Then there exists an open neighbourhood $U$ of the origin,
such that
\begin{itemize}\itemsep=0pt
\item[$(i)$]
$\overline{U} \subset \mathbb{D}$,

\item[$(ii)$]
$f_{t} \vert_{U}$ is univalent for each $t \in [0,T]$,

\item[$(iii)$]
$V := \bigcap_{0 \leqslant t \leqslant T} f_{t} ( U )$
is an open neighbourhood of the origin,

\item[$(iv)$]
for each $\zeta \in V$,
$[0,T] \ni t \mapsto f_{t}^{-1} (\zeta)$
is continuous and of bounded variation,

\item[$(v)$]
for each $\zeta \in V$,
with
$
f^{-1} ( t, \zeta )
:=
f_{t}^{-1} ( \zeta )
$
and
$
\mathrm{d}
f_{t}^{-1} ( \zeta )
:=
f^{-1} ( \mathrm{d}t, \zeta )
$,
we have
\begin{equation*}
\mathrm{d}
f_{t}^{-1} ( \zeta )
=
- f_{t}^{-1} ( \zeta )
\left\{
	\mathrm{d} x_{0} (t)
	+
	\sum_{k=1}^{\infty}
	\big( f_{t}^{-1} ( \zeta ) \big)^{k}
	\mathrm{d} x_{k} (t)
\right\},
\end{equation*}
as Lebesgue--Stieltjes measures on $[0,T]$.
\end{itemize}
\end{Prop} Let $\{ f_{t} \}_{0 \leqslant t \leqslant T}$ be a solution
to the controlled Loewner--Kufarev equation
\eqref{controlled-LK}.
Because of Corollary \ref{loc-univ},
associated to each $f_{t} (z)$ are the corresponding Faber polynomials
and Grunsky coefficients,
which will be denoted by
$Q_{n}(t,w)$,
and
$b_{-n,-m} (t)$,
respectively.

\begin{Prop}\label{Faber-variation} \quad
\begin{itemize}\itemsep=0pt
\item[$(i)$]
{\rm (Variation of Faber polynomials):}
We have for each $n \in \mathbb{N}$,
\begin{equation*}
\mathrm{d} Q_{n}( t, w )=n \mathrm{d} x_{n} (t)+n\sum_{k=1}^{n}Q_{k} (t,w) \mathrm{d} x_{n-k} (t) .
\end{equation*}

\item[$(ii)$]
{\rm (Variation of Grunsky coefficients):}
For each $n, m \in \mathbb{N}$,
\begin{gather}
\mathrm{d} b_{-n,-m} (t) =
- \mathrm{d} x_{n+m} (t)
+
\sum_{\substack{
k,l \in \mathbb{Z}_{\geqslant 0} ;\\
k+l = m-1
}}
( k+1 ) b_{-n,-(k+1)} (t)
\mathrm{d} x_{l} (t) \nonumber\\
\hphantom{\mathrm{d} b_{-n,-m} (t) =}{} +
\sum_{\substack{
k,l \in \mathbb{Z}_{\geqslant 0}; \\
k+l = n-1
}}
(k+1) b_{-m,-(k+1)}
\mathrm{d} x_{l} (t),\label{Grunsky_evol}
\end{gather}
with the initial condition
$b_{-n,-m} (0) = 0$,
for all $n,m \in \mathbb{N}$.
\end{itemize}
\end{Prop}
\begin{proof} (i)
Let $n \in \mathbb{N}$.
Let $U$ and $V$ be as in Proposition~\ref{inverse-evol}.
Then
$f_{t}^{-1} ( \zeta )$,
$\zeta \in V$,
satisfies
the equation
\begin{equation*}
\mathrm{d}
f_{t}^{-1} ( \zeta )
=
- f_{t}^{-1} ( \zeta )
\left\{
	\mathrm{d} x_{0} (t)
	+
	\sum_{k=1}^{\infty}
	\big( f_{t}^{-1} ( \zeta ) \big)^{k}
	\mathrm{d} x_{k} (t)
\right\} .
\end{equation*}
Let $X_{0} \subset V$ be an open disc centred at $0$.
By using Cauchy's integral formula,
we have for $w \in X_{0}$,
\begin{align*}
\mathrm{d} Q_{n}( t, w )&=
\frac{ 1 }{ 2 \pi i }
\int_{ \partial X_{0} }
\frac{ \mathrm{d} \zeta }{ \zeta }
\frac{
	\mathrm{d}
	\big( f_{t}^{-1} ( \zeta ) \big)^{-n}
}{ 1 - \zeta w^{-1} }\\
&=
\frac{ 1 }{ 2 \pi i }
\int_{ \partial X_{0} }
(-n)
\frac{ \big( f_{t}^{-1} ( \zeta ) \big)^{-n-1} }{ 1 - \zeta w^{-1} }
\big( {-} f_{t}^{-1} ( \zeta ) \big)
	\sum_{k=0}^{\infty}
	\big( f_{t}^{-1} ( \zeta ) \big)^{k}
\mathrm{d} x_{k} (t)
\frac{ \mathrm{d} \zeta }{ \zeta } \\
&=
\sum_{k=0}^{n}
\frac{ n }{ 2 \pi i }
\left(
\int_{ \partial X_{0} }
\frac{ \big( f_{t}^{-1} ( \zeta ) \big)^{-n+k} }{ 1 - \zeta w^{-1} }
\frac{ \mathrm{d} \zeta }{ \zeta }
\right)
\mathrm{d} x_{k} (t) \\
&=
\frac{
	n
	\mathrm{d} x_{n} (t)
}{ 2 \pi i }
\int_{ \partial X_{0} }
\frac{ 1 }{ 1 - \zeta w^{-1} }
\frac{ \mathrm{d} \zeta }{ \zeta }
+
n
\sum_{k=0}^{n-1}
Q_{n-k} (t,w)
\mathrm{d} x_{k} (t) .
\end{align*}
By noting that the orientation of
$
\partial X_{0}
$
is anti-clockwise,
we get
\begin{equation*}
\frac{ 1 }{ 2 \pi i }
\int_{ \partial X_{0} }
\frac{ 1 }{ 1 - \zeta w^{-1} }
\frac{ \mathrm{d} \zeta }{ \zeta }
= 1,
\end{equation*}
and hence the result.

(ii)
By putting
$
p( \mathrm{d}t ,z)
:=
\mathrm{d} x_{0} (t)
+
\mathrm{d} \xi ( \mathbf{x}, z )_{t}
$,
and since $f_{t}(z)$ satisfies
the controlled Loewner--Kufarev equation,
we have
\begin{align*}
\mathrm{d} Q_{n} ( t, f_{t} (z) )
&=
Q_{n}
( \mathrm{d} t, f_{t} (z) )
+
Q_{n}^{\prime} ( t, f_{t} (z) )
\mathrm{d} f_{t} (z) \\
&=
Q_{n}
( \mathrm{d} t, f_{t} (z) )
+
Q_{n}^{\prime} ( t, f_{t} (z) )
\big\{
	z f_{t}^{\prime} (z) p ( \mathrm{d}t , z )
\big\} \\
&=
Q_{n}
( \mathrm{d} t, f_{t} (z) )
+
z
\big[
\partial_{z} Q_{n} ( t, f_{t} (z) )
\big]
p ( \mathrm{d} t, z ) ,
\end{align*}
so that
\begin{equation}\label{t-der}
\mathrm{d} Q_{n} ( t, f_{t} (z) )
=
Q_{n}
( \mathrm{d} t, f_{t} (z) )
+
z
\big[
\partial_{z} Q_{n} ( t, f_{t} (z) )
\big]
p ( \mathrm{d} t, z ) .
\end{equation}
By recalling that
$Q_{n} ( t, f_{t} (z) )=z^{-n}+n
\sum_{m=1}^{\infty}b_{-n,-m} (t) z^{m}$, we have, by substitution, the following sequence of identities
\begin{equation}\label{LHS_pos}
(\text{LHS of (\ref{t-der})})_{\geqslant 1}
= (\text{LHS of (\ref{t-der})})
=
n
\sum_{m=1}^{\infty}
z^{m}
\mathrm{d} b_{-n,-m} (t) .
\end{equation}
Here, $(\cdots)_{\geqslant 1}$ is the operator which
forgets those terms in $(\cdots)$, whose degree is less than one.
On the other hand,
by Proposition~\ref{Faber-variation}(i),
we have
\begin{align*}
\mathrm{d} Q_{n}
( t, f_{t} (z) )&
=
n \mathrm{d} x_{n} (t)
+
n
\sum_{k=1}^{n}
Q_{k} ( t, f_{t} (z) )
\mathrm{d} x_{n-k} (t) \\
&=
n \mathrm{d} x_{n} (t)
+
n
\sum_{k=1}^{n}
\left(
	z^{-k}
	+
	k \sum_{m=1}^{\infty} b_{-k,-m} (t) z^{m}
\right)
\mathrm{d} x_{n-k} (t) \\
&=
n \mathrm{d} x_{n} (t)
+
n
\sum_{k=1}^{n}
z^{-k}
\mathrm{d} x_{n-k} (t)
+
n
\sum_{m=1}^{\infty}
\left(
\sum_{k=1}^{n}
k b_{-k,-m} (t)
\mathrm{d} x_{n-k} (t)
\right)
z^{m},
\end{align*}
so that
\begin{equation}\label{Q_n(f)-var}
\big(
\mathrm{d} Q_{n}
( t, f_{t} (z) )
\big)_{\geqslant 1}
=
n
\sum_{m=1}^{\infty}
\left(
\sum_{k=1}^{n}
k b_{-k,-m}
\mathrm{d} x_{n-k} (t)
\right)
z^{m}.
\end{equation}
We further have
\begin{gather*}
z
\big[
	\partial_{z} Q_{n} ( t, f_{t} (z) )
\big]
p ( \mathrm{d} t, z) =
z
\left(
	-n z^{-n-1}
	+
	n
	\sum_{k=1}^{\infty}
	k b_{-n,-k} z^{k-1}
\right)
\left(
	\mathrm{d} x_{0} (t)
	+
	\sum_{l=1}^{\infty}
	\mathrm{d} x_{l} (t) z^{l}
\right) \\
\hphantom{z\big[\partial_{z} Q_{n} ( t, f_{t} (z) )\big]p ( \mathrm{d} t, z)}{} =
n
\Bigg({-} \mathrm{d} x_{0} (t) z^{-n}
	- \sum_{m=1-n}^{\infty} \mathrm{d} x_{m+n} (t) z^{m} \\
\hphantom{z\big[\partial_{z} Q_{n} ( t, f_{t} (z) )\big]p ( \mathrm{d} t, z)=}{} +
	\sum_{m=1}^{\infty}
	m b_{-n,-m} (t) \mathrm{d} x_{0} (t) z^{m}
	+
	\sum_{m=2}^{\infty}
	\sum_{\substack{
	k,l \geqslant 1; \\
	k+l=m
	}}
	k b_{-n,-k} (t)
	\mathrm{d} x_{l} (t)
	z^{m}
\Bigg) ,
\end{gather*}
from which we conclude
\begin{equation}\label{rem_var}
\big(
z
\big[
\partial_{z} Q_{n} ( t, f_{t} (z) )
\big]
p ( \mathrm{d} t, z )
\big)_{\geqslant 1}
=
n
\sum_{m=1}^{\infty}
\left(
	- \mathrm{d} x_{n+m} (t)
	+
	\sum_{\substack{
	k \geqslant 1,\ l \geqslant 0 ; \\
	k+l = m
	}}
	k b_{-n,-k} (t)
	\mathrm{d} x_{l} (t)
\right)
z^{m} .
\end{equation}
Combining
(\ref{Q_n(f)-var})
and
(\ref{rem_var}),
we obtain
\begin{gather*}
(\text{RHS of (\ref{t-der})})_{\geqslant 1} =
n\sum_{m=1}^{\infty}
\Bigg(
	{-} \mathrm{d} x_{n+m} (t)
	+
	\sum_{\substack{
	k,l \in \mathbb{Z}_{\geqslant 0} ;\\
	k+l = m-1
	}}
	( k+1 ) b_{-n,-(k+1)} (t)
	\mathrm{d} x_{l} (t) \\
\hphantom{(\text{RHS of (\ref{t-der})})_{\geqslant 1} =}{} +
	\sum_{\substack{
	k,l \in \mathbb{Z}_{\geqslant 0}; \\
	k+l = n-1
	}}
	(k+1) b_{-m,-(k+1)}
	\mathrm{d} x_{l} (t)
\Bigg)
z^{m},
\end{gather*}
and then by comparing with~(\ref{LHS_pos}),
we get the result. Furthermore, the initial condition is derived from
$f_{0} (z) \equiv z$.
\end{proof}

In order to derive an explicit formula for the Grunsky coefficients~$b_{-n,-m}(t)$, cf.\ equation~(\ref{Grunsky_coefficients}),
we shall introduce some notation. In \cite{AFb}, we study analytic aspects of these coefficients.
\begin{Def}
Let $p,q \in \mathbb{N}$.
\begin{itemize}\itemsep=0pt
\item[(1)]
A bijection
$\sigma \colon \{ 1, 2, \ldots , p+q \}
\to \{ 1, 2, \ldots , p+q \}$
is called a {\it $(p,q)$-shuffle}
if it holds that
$\sigma (1) < \sigma (2) < \cdots < \sigma (p)$
and
$\sigma (p+1) < \sigma (p+2) < \cdots < \sigma (p+q)$.

\item[(2)]
Suppose that
$x_{1}, x_{2}, \ldots , x_{p+q}: [0,T] \to \mathbb{C}$
are continuous and of bounded variation.
Then for each $0 \leqslant t \leqslant T$,
we set
\begin{gather*}
\big(
(
	x_{1}
	\cdots
	x_{p}
)
\shuffle
(
	x_{p+1}
	\cdots
	x_{p+q}
)
\big)
(t) \\
\qquad{} :=:
\int\displaylimits_{
	0
	\leqslant s_{q}
	\leqslant \cdots
	\leqslant s_{1}
	\leqslant t_{p}
	\leqslant \cdots
	\leqslant t_{1}
	\leqslant t
}
\!\!\!\!\!\!\!\!\!\!\!\!\!\!\!\big(
	\mathrm{d} x_{1} (t_{1})
	\cdots
	\mathrm{d} x_{p} (t_{p})
\big)
\shuffle
\big(
	\mathrm{d} x_{p+1} (s_{1})
	\cdots
	\mathrm{d} x_{p+q} (s_{q})
\big) \\
\qquad{} :=
\sum_{
	\text{$\sigma^{-1}$: $(p,q)$-shuffle}
}
\int_{0}^{t} \mathrm{d} x_{\sigma (1)} (t_{1})
\int_{0}^{t_{1}} \mathrm{d} x_{\sigma (2)} (t_{2})
\cdots
\int_{0}^{t_{p-1}} \mathrm{d} x_{\sigma (p)} (t_{p}) \\
\qquad\quad{}\times
\int_{0}^{t_{p}} \mathrm{d} x_{\sigma (p+1)} (s_{1})
\int_{0}^{s_{1}} \mathrm{d} x_{\sigma (p+2)} (s_{2})
\cdots
\int_{0}^{s_{q-1}} \mathrm{d} x_{\sigma (p+q)} (s_{q}) .
\end{gather*}
\end{itemize}
\end{Def}

The general formula for the Grunsky-coefficients along the controlled Loewner--Kufarev equation~(\ref{controlled-LK})
is stated as next, and which is crucial for the embedding into the Grassmannian, cf.\ Section~\ref{Sec_LK/Gr}.
The proof is given in
Appendix~\ref{app:Grunsky-explicit}.

\begin{Prop}\label{explicit-Grunsky}
For $n,m \in \mathbb{N}$ and $t \geqslant 0$,
\begin{gather}
b_{-m,-n} (t) =
- \mathrm{e}^{ (n+m) x_{0}(t) }
\int_{0}^{t}
\mathrm{e}^{ - (n+m) x_{0}(s) }
\mathrm{d} x_{m+n} (s) \nonumber\\
\hphantom{b_{-m,-n} (t) =}{} -
\sum_{k=2}^{n+m-2}
\sum_{\substack{
	1 \leqslant i < m; \\
	1 \leqslant j < n: \\
	i+j=k
}}
\sum_{p=1}^{m-i}
\sum_{q=1}^{n-j}
\sum_{\substack{
	i_{1}, \ldots , i_{p} \in \mathbb{N}\colon \\
	i_{1} + \cdots + i_{p} = m-i
}}
\sum_{\substack{
	j_{1}, \ldots , j_{q} \in \mathbb{N}\colon \\
	j_{1} + \cdots + j_{q} = n-j
}}
w(i,j)_{i_{1}, \ldots , i_{p}; j_{1}, \ldots , j_{q}} \nonumber\\
\hphantom{b_{-m,-n} (t) =}{}\times
\mathrm{e}^{(m+n)x_{0}(t)}
\int\displaylimits_{
	0
	\leqslant u_{q}
	\leqslant \cdots
	\leqslant u_{1}
	\leqslant s_{q}
	\leqslant \cdots
	\leqslant s_{1}
	\leqslant t
}
\big(
	\mathrm{e}^{-i_{1}x_{0}(s_{1})}
	\mathrm{d} x_{i_{1}} (s_{1})
	\cdots
	\mathrm{e}^{-i_{p}x_{0}(s_{p})}
	\mathrm{d} x_{i_{p}} (s_{p})
\big) \nonumber\\
\hphantom{b_{-m,-n} (t) =}{} \shuffle
\big(
	\mathrm{e}^{-j_{1}x_{0}(u_{1})}
	\mathrm{d} x_{j_{1}} (u_{1})
	\cdots
	\mathrm{e}^{-j_{q}x_{0}(u_{q})}
	\mathrm{d} x_{j_{q}} (u_{q})
\big)
\int_{0}^{u_{q}}
\mathrm{e}^{ - k x_{0} (s) }
\mathrm{d} x_{k} (s) \nonumber\\
\hphantom{b_{-m,-n} (t) =}{} -
\sum_{k=m+1}^{n+m-1}
\sum_{q=1}^{n+m-k}
\sum_{\substack{
	j_{1}, \ldots , j_{q} \in \mathbb{N}\colon \\
	j_{1} + \cdots + j_{q} = n+m-k
}}
w(k-m)_{ \varnothing ; j_{1}, \ldots , j_{q} } \nonumber\\
\hphantom{b_{-m,-n} (t) =}{} \times
\mathrm{e}^{(m+n)x_{0}(t)}
\int\displaylimits_{
	0
	\leqslant s_{q}
	\leqslant \cdots
	\leqslant s_{1}
	\leqslant t
}
\big(
	\mathrm{e}^{-j_{1}x_{0}(s_{1})}
	\mathrm{d} x_{j_{1}} (s_{1})
	\cdots
	\mathrm{e}^{-j_{q}x_{0}(s_{q})}
	\mathrm{d} x_{j_{q}} (s_{q})
\big)\nonumber\\
\hphantom{b_{-m,-n} (t) =}{}\times
\int_{0}^{s_{q}}
\mathrm{e}^{ - k x_{0} (s) }
\mathrm{d} x_{k} (s)
-
\sum_{k=n+1}^{n+m-1}
\sum_{p=1}^{m+n-k}
\sum_{\substack{
	i_{1}, \ldots , i_{p} \in \mathbb{N}\colon \\
	i_{1} + \cdots + i_{p} = m+n-k
}}
w(k-n)_{ i_{1}, \ldots , i_{p}; \varnothing } \nonumber\\
 \hphantom{b_{-m,-n} (t) =}{} \times
\mathrm{e}^{(m+n)x_{0}(t)}
\int \displaylimits_{
	0
	\leqslant u_{p}
	\leqslant \cdots
	\leqslant u_{1}
	\leqslant t
}
\!\!\!\!\!\!\big(
	\mathrm{e}^{-i_{1}x_{0}(u_{1})}
	\mathrm{d} x_{i_{1}} (u_{1})
	\cdots
	\mathrm{e}^{-i_{p}x_{0}(u_{p})}
	\mathrm{d} x_{i_{p}} (u_{p})
\big)\nonumber\\
\hphantom{b_{-m,-n} (t) =}{}\times
\int_{0}^{u_{p}}
\mathrm{e}^{ - k x_{0} (u) }
\mathrm{d} x_{k} (u),\label{Grunsky-formula}
\end{gather}
where, for
$m=i_{1}+ \cdots + i_{p} + r$,
and
$n=j_{1}+ \cdots + j_{q} + s$,
we have put
\begin{gather*}
w(r)_{i_{1},\ldots ,i_{p}; \varnothing}
=
( m - i_{1} )
( m - (i_{1}+i_{2}) )
\cdots
( m - (i_{1}+i_{2}+\cdots +i_{p}) ) , \\
w(s)_{\varnothing ; j_{1}, \ldots , j_{q}}
=
( n - j_{1} )
( n - (j_{1}+j_{2}) )
\cdots
( n - (j_{1}+j_{2}+\cdots +j_{q}) ) ,
\end{gather*}
and
$w(r,s)_{i_{1},\ldots ,i_{p}; j_{1}, \ldots , j_{q}}:=
w(r)_{i_{1}, \ldots , i_{p}; \varnothing}w(s)_{\varnothing ; j_{1}, \ldots , j_{q}}$.
\end{Prop}

\section[The controlled Loewner--Kufarev equation embedded into the Segal--Wilson Grassmannian]{The controlled Loewner--Kufarev equation embedded\\ into the Segal--Wilson Grassmannian}\label{Sec_LK/Gr}

\subsection{Segal--Wilson Grassmannian}

Let $H := L^{2} \big( S^{1},\mathbb{C} \big)$
be the Hilbert space which consists of all
square-integrable complex functions
on the unit circle~$S^{1}$.
It decomposes orthogonally into
$H = H_{+} \oplus H_{-}$,
where
$H_{+}$
and
$H_{-}$
are the closure of
$\operatorname{span} \big\{ z^{k}\colon k \geqslant 0 \big\}$
and
$\operatorname{span} \big\{ z^{k}\colon k < 0 \big\}$,
respectively.

\begin{Def}[{G.~Segal and G.~Wilson~\cite[Section~2]{SW}}]
The {\it Segal--Wilson Grassmannian}
$\mathrm{Gr} := \mathrm{Gr} (H)$
is the set of all closed subspaces $W$ of $H$
satisfying the following:
\begin{itemize}\itemsep=0pt
\item[(1)]
The orthogonal projection
$\mathrm{pr}_{+}\colon W \to H_{+}$
is Fredholm,

\item[(2)]
The orthogonal projection
$\mathrm{pr}_{-}\colon W \to H_{-}$
is compact.
\end{itemize}
The Fredholm index of the orthogonal projection
$\mathrm{pr}_{+}\colon W \to H_{+}$
is called the
{\it virtual dimension}
of $W$.
For $d \in \mathbb{Z}$, we set
\begin{equation*}
\mathrm{Gr}
\big({\textstyle \frac{\infty}{2}+d}, \infty\big)
:=
\{ W \in \mathrm{Gr}\colon \text{the virtual dimension of $W$ is $d$} \},
\end{equation*}
and
$\mathrm{Gr}\big({\textstyle \frac{\infty}{2}}, \infty\big):=\mathrm{Gr}\big({\textstyle \frac{\infty}{2}+0}, \infty\big)$.
\end{Def}
If we take $W=H_{+}$,
then the corresponding projections are given by
$\mathrm{pr}_{+} = \mathrm{id}_{H_{+}}$
and
$\mathrm{pr}_{-} = 0$,
which are Fredholm and compact operators, respectively.
Therefore we have
$H_{+} \in \mathrm{Gr} \big( \frac{\infty}{2}, \infty \big)$.

\begin{Def}[{\cite[Section~5]{SW}}]
Let $\Gamma_{+}$ denote
the set of all continuous functions
$g\colon S^{1} \to \mathbb{C}^{*}$, such that
$g (z)=\mathrm{e}^{\text{{\tiny $\sum_{k=1}^{\infty} t_{k} z^{k}$}}}$,
$z \in S^{1}$
for some $\mathbf{t} = ( t_{1}, t_{2}, t_{3}, \ldots )$.
\end{Def}

The set $\Gamma_{+}$ acts on $H$ by pointwise multiplication.
In particular, $\Gamma_{+}$ forms a group.
This action induces the action of $\Gamma_{+}$
on
$\mathrm{Gr}\colon \Gamma_{+} \times \mathrm{Gr}\ni (g,W)\mapsto
gW \in \mathrm{Gr}$
(see \cite[Lemma~2.2 and Proposition~2.3]{SW}),
where $gW = \{ g f\colon f \in W \}$.
For any
$
g = \mathrm{e}^{\text{{\tiny $\sum_{k=1}^{\infty} t_{k} z^{k}$}}}
\in \Gamma_{+}
$,
the action of $g$ on $H$
is of the form
\begin{equation*}
g
=
\left(\begin{matrix}
a & b \\
0 & d
\end{matrix}\right)
\qquad
\text{along $H = H_{+} \oplus H_{-}$,}
\end{equation*}
where
$a\colon H_{+} \to H_{+}$
is invertible
and
$b\colon H_{-} \to H_{+}$
is of trace class
(see \cite[Proposition~2.3]{SW}).
Let $\mathcal{U}$ be the set of all
$W \in \mathrm{Gr} \big({\textstyle \frac{\infty}{2}}, \infty\big)$
such that
the orthogonal projection $W \to H_{+}$
is an isomorphism.
Then, associated to each $W \in \mathcal{U}$
is the
{\it tau-function}
$\tau_{W} ( \mathbf{t} )$
of~$W$,
a function of infinitely many ``times'' $\mathbf{t} = ( t_{1}, t_{2}, \ldots )$. It is known that the following holds:

\begin{Prop}[{\cite[Proposition~3.3]{SW}}]
Let $W \in \mathcal{U}$.
For
$g=\mathrm{e}^{\text{{\tiny $\sum_{n=1}^{\infty} t_{n} z^{n}$}}}
\in \Gamma_{+}$,
we have
\begin{equation*}
\tau_{W} ( \mathbf{t} )
=
\mathrm{det} \big( 1+a^{-1} b A \big),
\end{equation*}
where
$\mathbf{t} = ( t_{1}, t_{2}, t_{3}, \ldots )$,
\begin{equation*}
g^{-1}
=
\left(\begin{matrix}
a & b \\ 0 & d
\end{matrix}\right)
\qquad
\text{along $H=H_{+} \oplus H_{-}$},
\end{equation*}
and $A\colon H_{+} \to H_{-}$ is
the linear operator such that $\mathrm{graph} (A) = W$.
\end{Prop}

\subsection{Krichever's construction}\label{Krichever_Constr}

In connection with algebraic geometry
and infinite-dimensional integrable systems,
a fundamental observation / construction of Krichever \cite{Kr1, Kr2, Kr3}
states the following.
A solution of the KdV equation is associated with each
non-singular algebraic curve, equipped with some additional
algebro-geometric data.
Segal and Wilson~\cite{SW} developed and formalised, after a remark by Mumford~\cite{Mu}, this construction further.

The specific algebro-geometric datum
is given by a quintuple $( X, \EuScript{L}, x_{\infty}, z, \varphi )$, consisting of the following parts. $X$ is a complete, irreducible and complex algebraic curve with
a rank-one, torsion-free coherent sheaf $\EuScript{L}$.
Additionally, a non-singular point $x_{\infty} \in X$, and a closed neighbourhood $X_{\infty}$, are chosen, such that there exists a
local parameter
$1/z \colon X_{\infty}\to\overline{\mathbb{D}}
\subset \widehat{\mathbb{C}}$,
with $x_{\infty}\mapsto0$,
and a trivialisation
$\varphi \colon
\EuScript{L}\vert_{X_{\infty}}\to
\overline{\mathbb{D}} \times \mathbb{C}$,
of $\EuScript{L}\vert_{X_{\infty}}$.
Each section of $\EuScript{L}\vert_{X_{\infty}}$ is
identified with a complex function on $\overline{\mathbb{D}}$
under~$\varphi$.
For $X_{0} := X \setminus \interior{X}_{\infty}$, with $ \interior{X}_{\infty}$ the interior of $ X_{\infty}$,
the closed sets
$X_{0}$ and $X_{\infty}$ cover $X$,
and
$X_{0} \cap X_{\infty}$
is identified with $S^{1}$ under $z$.

Given this algebro-geometric datum,
one can associate a closed subspace
$W \subset H$, consisting
of all analytic functions $S^{1} \to \mathbb{C}$
which, under the above identification,
extend to a holomorphic section of $\EuScript{L}$ on
an open neighbourhood of $X_{0}$.
More explicitly,
one can write
\begin{equation*}
W=
\overline{
\left\{
	\begin{matrix}
	\text{the second component} \\
	\text{of $\varphi \circ s \circ (1/z)^{-1}\vert_{S^{1}}$}
	\end{matrix}
\colon
	\begin{matrix}
	\text{$s$ is a holomorphic section} \\
	\text{on a neighbourhood of $X_{0}$}
	\end{matrix}
\right\}
}^{H} ,
\end{equation*}
where
$
(1/z)^{-1}\colon
\overline{\mathbb{D}}
\to
X_{\infty}
$
is the inverse function of $1/z$.
It is known that
$W \in \mathrm{Gr}$
(see \cite[Proposition~6.1]{SW}),
and if $X$ is a compact Riemann surface
(then $\EuScript{L}$ is automatically a complex line bundle,
hence a maximal torsion-free sheaf),
this correspondence
$( X, \EuScript{L}, x_{\infty}, z, \varphi ) \mapsto W \in \mathrm{Gr}$
is one-to-one
(see \cite[Proposition~6.2]{SW}).

\subsection{The appearance of Faber polynomials and Grunsky coefficients}\label{FabGrunsky}

Let
$f\colon \mathbb{D} \to \mathbb{C}$
be a univalent function
such that $f(0) = 0$, and
$f(\mathbb{D})$ is bounded by a Jordan curve.
We set
$
\beta \colon \widehat{\mathbb{C}} \to \widehat{\mathbb{C}}
$
by $\beta (w) := 1/w$.
For a subset $A \subset \widehat{\mathbb{C}}$,
we shall write $A^{-1} := \beta (A)$, and
let
$
\widehat{\mathbb{D}}_{\infty}
:=
\widehat{\mathbb{C}} \setminus \overline{\mathbb{D}}
$.
We obtain an algebro-geometric datum
$
( X, \EuScript{L}, x_{\infty}, z, \varphi )
$
by setting
$X= \widehat{\mathbb{C}}$,
$\EuScript{L} = \widehat{\mathbb{C}} \times \mathbb{C}$,
$
x_{\infty} := \infty
$,
$X_{\infty}:=f \big( \overline{\mathbb{D}} \big)^{-1}$, $z:= \beta \circ f^{-1} \circ\beta^{-1} \colon X_{\infty}\to\widehat{\mathbb{D}}_{\infty}$,
and
$\varphi = (1/z) \times \mathrm{id}_{\mathbb{C}}$.
Correspondingly, we have
$X_{0}= \widehat{\mathbb{C}}
\setminus \big( f( \mathbb{D} )^{-1} \big)$.
Further, by the Caratheodory extension theorem, $z$ extends continuously to
$X_{\infty}$, and therefore we can embed $f$, by assigning
a Hilbert space $W = W_{f}$ to it, into the Grassmannian.
In this case, we have $\widehat{\mathbb{C}}\setminus\big( f( \mathbb{D} )^{-1} \big)$, and hence
\begin{equation*}
W_{f}
=
\overline{
\left\{
F \circ ( 1/z )^{-1} \vert_{S^{1}} \colon
\begin{matrix}
\text{$F$ is a holomorphic function} \\
\text{on a neighbourhood of $
\widehat{\mathbb{C}}
\setminus
\big( f( \mathbb{D})^{-1} \big)
$}
\end{matrix}
\right\}
}^{H} .
\end{equation*}
In order to start this paper's main calculation,
let us specify this more explicitly.
For a closed subset $V$ in $\widehat{\mathbb{C}}$,
we denote by
$\mathcal{O} (V)$
the space of all holomorphic functions
defined on an open neighbourhood of~$V$.
For a univalent function
$
g\colon
\widehat{\mathbb{D}}_{\infty}
\to
\widehat{\mathbb{C}},
$
with $g(\infty ) = \infty$,
and for each
$h \in \mathcal{O} \big( \overline{\mathbb{D}} \big)$,
we call
\begin{equation*}
( \EuScript{F}[h]) ( z )
:=
\frac{1}{ 2 \pi i }
\int_{ \partial ( \mathbb{C} \setminus g( \mathbb{D}_{\infty} ) ) }
\frac{ h \big( g^{-1} ( \xi ) \big) }{ \xi - z }
\mathrm{d} \xi ,
\qquad
z \in \mathbb{C} \setminus \overline{g( \mathbb{D}_{\infty} )}
\end{equation*}
the
{\it Faber transform}
of~$h$
(with respect to~$g$).
If the boundary
$\partial ( \mathbb{C} \setminus g( \mathbb{D}_{\infty} ) )$
is analytic,
it is known that
$h \in \mathcal{O} \big( \overline{\mathbb{D}} \big)$
iff
$\EuScript{F} h \in \mathcal{O} ( \mathbb{C} \setminus g( \mathbb{D}_{\infty} ) )$
(see \cite[Theorem~1]{Jo})
and
$
\EuScript{F}\colon
\mathcal{O} \big( \overline{\mathbb{D}} \big)
\to
\mathcal{O} ( \mathbb{C} \setminus g( \mathbb{D}_{\infty} ) )
$
is bijective.
In our case,
we put
\begin{equation*}
g := ( 1/z )^{-1} = \beta \circ f \circ \beta^{-1}
\colon \ \widehat{\mathbb{D}}_{\infty}
\to f(\mathbb{D})^{-1},
\end{equation*}
and then we can describe
$\mathcal{O} ( X_{0} )$
by
$\mathcal{O} \big( \widehat{\mathbb{D}}_{\infty} \big)$
through the transformation
\begin{equation*}
\EuScript{F}
\circ \big( \beta^{-1} \big)^{*}
=
\big( \beta^{-1} \big)^{*}
\circ
\mathrm{Ad}_{\beta^{*}} ( \EuScript{F} )\colon \
\mathcal{O} \big( \widehat{\mathbb{D}}_{\infty} \big)
\to
\mathcal{O} ( X_{0} ),
\end{equation*}
where
$\mathrm{Ad}_{\beta^{*}} ( \EuScript{F} )
:=
\beta^{*}
\circ \EuScript{F}
\circ \big( \beta^{-1} \big)^{*} \colon
\mathcal{O} \big( \widehat{\mathbb{D}}_{\infty} \big)
\to
\mathcal{O} \big( \widehat{\mathbb{C}} \setminus f(\mathbb{D}) \big)
$.
A direct calculation shows that for each
$
h ( \eta ) = \sum_{k=0}^{\infty} a_{k} \eta^{-k}
\in
\mathcal{O} \big( \widehat{\mathbb{D}}_{\infty} \big)
$,
we have
\begin{equation*}
( \mathrm{Ad}_{\beta^{*}} ( \EuScript{F} ) [h] )( w )
=
\frac{1}{2\pi i}
\int_{ \partial f(\mathbb{D}) }
\frac{ h \big( f^{-1} ( \zeta ) \big) }{ 1 - \zeta w^{-1} }
\frac{ \mathrm{d} \zeta }{ \zeta },
\qquad
w \in \widehat{\mathbb{C}} \setminus f(\mathbb{D}).
\end{equation*}
As a result,
$( \mathrm{Ad}_{\beta^{*}} ( \EuScript{F} )[h])( w )$
is a power series in $1/w$.
Actually, in view of the Cauchy integral formula
\begin{equation*}
\frac{1}{2 \pi i}
\int_{S^{1}}
\frac{ \zeta^{n} }{ 1 - \zeta \eta^{-1} }
\frac{ \mathrm{d} \zeta }{ \zeta }
=
 \begin{cases}
\eta^{n} & \text{if $n \leqslant 0$}, \\
0 & \text{if $n \geqslant 1$},
\end{cases}
\qquad
\eta \in \mathbb{D}_{\infty},
\end{equation*}
we have
\begin{equation*}
\begin{split}
( \mathrm{Ad}_{\beta^{*}} ( \EuScript{F} )
[h]
)( w )
=
\sum_{k=0}^{n}
\frac{a_{k}}{2\pi i}
\int_{ \partial X_{0} }
\frac{ \big( f^{-1} ( \zeta ) \big)^{-k} }{ 1- \zeta w^{-1} }
\frac{ \mathrm{d} \zeta }{ \zeta }
=
\sum_{k=0}^{n}
a_{k}
\big[ \big( f^{-1} (w) \big)^{-k} \big]_{\leqslant 0},
\end{split}
\end{equation*}
where
$\big[ \big( f^{-1} (w) \big)^{-k} \big]_{\leqslant 0}$
denotes the
constant-part
plus the
principal-part
of the Laurent series for
$\big( f^{-1} (w) \big)^{-k} = \big( 1 / f^{-1} (w) \big)^{k}$; hence every element in
$\mathcal{O} \big( \widehat{\mathbb{C}} \setminus f(\mathbb{D}) \big)$
can be written as a series in~$1/w$.
The quantity
\begin{equation}\label{Faber_alt}
Q_k(w):=\frac{1}{2\pi i}\int_{ \partial X_{0} }
\frac{ \big( f^{-1} ( \zeta ) \big)^{-k} }{ 1- \zeta w^{-1} }
\frac{ \mathrm{d} \zeta }{ \zeta }=\big[ \big( f^{-1} (w) \big)^{-k} \big]_{\leqslant 0},
\end{equation}
for $k\in\mathbb{N}$, is called the {\it $k$-th Faber polynomial} associated to the domain
$\mathbb{C} \setminus \overline{f(\mathbb{D})}$
(or simply to~$f$), and it is a polynomial of degree $k$ in $1/w$, cf.\ also~Section~{\ref{Grunsky}}.\\
We conclude that
$\big[ \big(\beta^{-1}\big)^{*} \circ \mathrm{Ad}_{\beta^{*}} (h) \big]
\circ(1/z)^{-1}
=[ \mathrm{Ad}_{\beta^{*}} (h) ]
\circ f \circ \beta^{-1}$, and hence
\begin{equation*}
W_{f}
=
\overline{
	\operatorname{span}
	\big(
	\{ 1 \}
	\cup
	\{
		Q_{n} \circ f \circ (1/z) \vert_{S^{1}}
	\}_{n \geqslant 1}
	\big)
}^{H} ,
\end{equation*}
where $z$ is the identity map on
$\widehat{\mathbb{D}}_{\infty}$; note, if $f(z) \equiv z$ then $W_{f} = H_{+}$.

\begin{Rm} \quad
\begin{itemize}\itemsep=0pt
\item[(a)]
The Faber polynomials appeared first
(with a different formalism,
but equivalent to our presentation)
in the context
of approximations of functions in one complex variable
by analytic functions
(see~\cite{El} and~\cite{Fa}).
Since then, they also play an important role
in the theory of univalent functions
(see~\cite{Sc}).
We introduced the Faber polynomials in a slightly
non-standard way in order to have them in a form which is suitable for embedding univalent functions into the Grassmannian
by using Faber polynomials.

\item[(b)]
In the context of Abelian function theory,
the exterior derivatives
\begin{equation*}
\omega_{\infty}^{(n)} := \mathrm{d} Q_{n} (f(1/z)),
\end{equation*}
$n=1,2,\ldots$
are known as Abelian differentials of the second kind on the Riemann sphere.
In general, Krichever's embedding of the algebro-geometric datum
$(X, \EuScript{O}, Q, z, \varphi )$,
where
\begin{equation*}
(X, \alpha_{1}, \ldots , \alpha_{g}, \beta_{1}, \ldots , \beta_{g})
\end{equation*}
is a homologically marked compact Riemann surface with genus~$g$,
$\EuScript{O}$ is the structure sheaf of~$X$,
$Q \in X$,
$z$ and $\varphi$
are local uniformisers, and a local trivialisation of~$\EuScript{O}$,
is described by using
multivalued meromorphic functions
$\varphi^{(0)} (z) \equiv 1$,
\begin{equation*}
\varphi^{(n)} (z) := \int^{z} \omega_{Q}^{(n)}
=:
z^{n} - \sum_{m=1}^{\infty} q_{nm} \frac{z^{-m}}{m},
\end{equation*}
(modulo periods)
where
$\omega_{Q}^{(n)}$'s are
(normalised)
abelian differentials
of the second kind
\cite[Section~2.27 and p.~304]{KNTY}.
These multivalued meromorphic functions
can be regarded as a~generalisation of the Faber polynomials
(see \cite[p.~131]{Ya80}).

\item[(c)]
Given again a homologically marked compact Riemann surface
$\big(\!X, (\alpha_{i}, \beta_{i})_{i=1}^{g}\!\big)$
with genus~$g$,
Krichver's embedding of yet another datum
$\big(X, \Omega^{1/2}, Q, z, \sqrt{\mathrm{d}z} \big)$
or
\begin{equation*}
\big(X, \Omega^{1/2} \otimes \EuScript{L}_{c}, Q, z, \sqrt{\mathrm{d}z} \otimes s_{c} \big)
\end{equation*}
is described in \cite[equation~(2.34)]{KNTY}.
Here,
$\Omega^{1/2}$
is the so-called
{\it theta characteristic}
of the compact Riemann surface~$X$,
$\EuScript{L}_{c}$
is a complex line bundle of degree $0$ parametrised by
$c \in \mathbb{C}^{g}$
(modulo the lattice associated to $(\alpha_{i} , \beta_{i})_{i=1}^{g}$),
and
$s_{c}$ is a local trivialisation of $\EuScript{L}_{c}$.
In particular, the embedding of the latter
and the associated Fermionic state
(the image under the Pl\"{u}cker embedding)
are described by means of the Szeg\H{o} kernel of
$\Omega^{1/2} \otimes \EuScript{L}_{c}$
(see~\cite{AGMV, KNTY}, in which, the
{\it scattering operator}
in \cite[Section~5.12]{KNTY}
is a special case of a
{\it Bogoliubov transformation}
discussed in
\cite[equations~(2.15)--(2.20)]{AGMV}),
and then the corresponding tau-function
$\tau (\mathbf{t})$
is described as a
theta function multiplied by
$\exp \big( \sum_{n,m=1}^{\infty} q_{nm}t_{n}t_{m} \big)$
(see \cite[Theorem~5.6]{KNTY}).
\end{itemize}
\end{Rm}

\subsection{Action of words in Witt algebra generators}

Let $X=\{ x_{1}, x_{2}, x_{3}, \ldots \}$
be an alphabet, consisting of a
countable set of
non-commuting letters.
The free monoid~$X^{*}$ on~$X$ is the set of
all words in the letters~$X$, including the empty word $\varnothing$.
We denote by
\begin{equation*}
\mathbb{C} \langle X \rangle
:=
\bigoplus_{w \in X^{*}} \mathbb{C}w
=
\mathbb{C}
\oplus
\bigoplus_{n=1}^{\infty}
\mathbb{C} \langle X \rangle_{n}
\end{equation*}
the free associative and unital $\mathbb{C}$-algebra on $X$.
The unit of this algebra is the empty word which we will
denote by $1 := \varnothing$.
The set $\mathbb{C} \langle X \rangle_{n}$
stands for $\bigoplus_{|w|=n} \mathbb{C} w$ where the summation
is taken over all words $w$ of length $n$.

\begin{Def}
We define
\begin{equation*}
\xi ( \mathbf{x}, z )
:=
\sum_{n=1}^{\infty }
x_{k} z^{k}
\in
\mathbb{C} \langle X \rangle [\![ z ]\!],
\end{equation*}
and a distinguished element
$S( \xi ( \mathbf{x}, z ) )\in \mathbb{C} \langle X \rangle [\![ z ]\!]$
by
\begin{equation*}
S( \xi ( \mathbf{x}, z ) )
:=1+\sum_{n=1}^{\infty}
z^{n}
\sum_{p=1}^{n}
\sum_{
	\substack{
		i_{1}, \ldots , i_{p} \in \mathbb{N}\colon \\
		i_{1} + \cdots + i_{p} = n
	}
}
\!\!\!\!\!\! x_{i_{1}} \cdots x_{i_{p}} .
\end{equation*}
\end{Def}
\begin{Def}Let
$x_{0}\colon [0,+\infty ) \to \mathbb{R}$
and
$x_{1}, x_{2}, \ldots\colon [0,+\infty ) \to \mathbb{C}$
be continuous and of bounded variation.
For $0 \leqslant s \leqslant t$,
we define
$
[ \int 1 ]_{s,t}
:=
1
$
and
\begin{gather*}
\left[
	\int \!( x_{i_{p}} \cdots x_{i_{2}} x_{i_{1}} )
\right]_{s,t}\!
:=
\int \displaylimits_{s \leqslant u_{1} < u_{2} < \cdots < u_{p} \leqslant t}\!\!\!\!\!\!\!\!\!\!\!\!\!\!\!\!\!
\mathrm{e}^{ -i_{1} x_{0} (u_{1}) }
\mathrm{d} x_{i_{1}} (u_{1})
\mathrm{e}^{ -i_{2} x_{0} (u_{2}) }
\mathrm{d} x_{i_{2}} (u_{2})
\cdots
\mathrm{e}^{ -i_{p} x_{0} (u_{p}) }
\mathrm{d} x_{i_{p}} (u_{p}) .
\end{gather*}
The action of $\int$ naturally extends to
$
\mathbb{C} \langle X \rangle [\![ z ]\!],
$
and then we call
\begin{equation*}
\begin{split}
S( \xi (\mathbf{x}, z) )_{s,t}
:=
\left[
	\int S( \xi (\mathbf{x}, z) )
\right]_{s,t},
\end{split}
\end{equation*}
the {\it signature} of $\xi ( \mathbf{x}, z )$.
\end{Def}

We define a bilinear map
$
T\colon
\mathbb{C} \langle X \rangle \big(\!\big( z^{-1} \big)\!\big)
\times
\mathbb{C} \langle X \rangle
\to
\mathbb{C} \langle X \rangle \big(\!\big( z^{-1} \big)\!\big),
$
by extending the pairing $T(f,1):=f$, and
$
T(f,x_{i_{p}} \cdots x_{i_{1}}):=( L_{-i_{1}} \cdots L_{-i_{p}} f )
x_{i_{p}} \cdots x_{i_{1}},
$
bilinearly, for
$
f \in \mathbb{C} \langle X \rangle \big(\!\big( z^{-1} \big)\!\big)
$,
$p \geqslant 1$,
and $i_{1}, \ldots , i_{p} \in \mathbb{N}$.
Further,
$
L_{k} :=
- z^{k+1} \partial / (\partial z),
$
for $k \leqslant -1$,
forms the negative part of the Witt algebra, cf.~(\ref{Witt_alg}), and
$\partial / (\partial z)$
is a formal derivation on
$\mathbb{C} \langle X \rangle \big(\!\big( z^{-1} \big)\!\big)$.

For
$f \in \mathbb{C} \langle X \rangle \big(\!\big( z^{-1} \big)\!\big)$
and
$x \in \mathbb{C} \langle X \rangle$, in the sequel,
$T(f,x)$, will be denoted by~$f._{z}x$.
The following is clear by definition:

\begin{Prop}$T$ defines an action of the $\mathbb{C}$-algebra
$\mathbb{C} \langle X \rangle$
on $\mathbb{C} \langle X \rangle \big(\!\big( z^{-1} \big)\!\big)$
from the right.
\end{Prop}

The right action $T$ can be extended to the right action
\begin{equation}\label{action1}
\mathbb{C} \langle X \rangle \big(\!\big( w^{-1} \big)\!\big)
\times
\mathbb{C} \langle X \rangle [\![ z ]\!]
\to
\mathbb{C} \langle X \rangle \big(\!\big( w^{-1} \big)\!\big) [\![ z ]\!],
\end{equation}
under which the image of
$
( f, z^{n} x_{i_{p}} \cdots x_{i_{1}} )
$
is mapped to
$
z^{n} ( f._{w} x_{i_{p}} \cdots x_{i_{1}} )
=:
f._{w} ( z^{n} x_{i_{p}} \cdots x_{i_{1}} )
$.
Note that now the notation
$f ._{w} S(\mathbf{x})$
makes sense.

\begin{Thm} 
\label{Sol-by-Witt} 

Let $\{ f_{t} \}_{t \geqslant 0}$ be a solution
to the Loewner--Kufarev equation.
Then
\begin{equation*}
f_{t} (z)
=
\left[
	\int
	\underset{w=0}{\operatorname{Res\ }}
	\left(
	\frac{ \mathrm{e}^{x_{0}(t)} z }{ 1 - z w }
	\big(
		w^{-1} ._{w} S\big( \xi \big( \mathbf{x} , \mathrm{e}^{x_{0}(t)} \big) \big)
	\big)
	\right)
\right]_{0,t}.
\end{equation*}
\end{Thm}
\begin{proof}
By setting
\begin{equation*}
\widetilde{w}(n)_{i_{1},\ldots , i_{p}}
:=
\big\{ (n-i_{1}) + 1 \big\}
\big\{ (n-(i_{1}+i_{2})) + 1 \big\}
\cdots
\big\{ \big( n - ( i_{1} + i_{2} + \cdots + i_{p-1} ) \big) + 1 \big\},
\end{equation*}
where $n = i_{1} + \cdots + i_{p}$, we have
\begin{gather*}
w^{-1} ._{w} 1 = w^{-1}, \\
w^{-1} ._{w} x_{i_{p}} \cdots x_{i_{1}}=
\widetilde{w}(n)_{i_{1},\ldots , i_{p}}
x_{i_{p}} \cdots x_{i_{1}}
w^{ - ( i_{1} + \cdots + i_{p} + 1 ) }.
\end{gather*}
Therefore
$\underset{w=0}{\operatorname{Res\ }}\big(
	\sum_{m=0}^{\infty}
	z^{m} w^{m}
	\big( w^{-1} ._{w} 1 \big)
\big)
=
1
~(\text{i.e., the empty word $\varnothing$} )
$,
and
\begin{equation*}
\underset{w=0}{\operatorname{Res\ }}
\left(
	\sum_{m=0}^{\infty}
	z^{m} w^{m}
	\big( w^{-1} ._{w} x_{i_{p}} \cdots x_{i_{1}} \big)
\right)
=
z^{ ( i_{1} + \cdots + i_{p} ) }
\widetilde{w}(n)_{i_{1},\ldots , i_{p}}
x_{i_{p}} \cdots x_{i_{1}} .
\end{equation*}
Hence we get
\begin{gather*}
\underset{w=0}{\operatorname{Res\ }}
\left(
	\frac{ \mathrm{e}^{ x_{0} (t) } z }{ 1 - z w }
	\big(
		w^{-1} ._{w} S\big( \xi \big( \mathbf{x}, \mathrm{e}^{ x_{0} (t) } \big) \big)
	\big)
\right) \\
\qquad{} =
\mathrm{e}^{ x_{0} (t) }
z
+
\sum_{n=1}^{\infty}
\mathrm{e}^{ (n+1) x_{0} (t) }
z^{ n+1 }
\sum_{p=1}^{n}
\sum_{
	\substack{
		i_{1}, \ldots , i_{p} \in \mathbb{N}\colon \\
		i_{1} + \cdots + i_{p} = n
	}
}
\widetilde{w}(n)_{i_{1},\ldots , i_{p}}
x_{i_{p}} \cdots x_{i_{1}} .
\end{gather*}
Now, in view of Theorem \ref{explicit-Taylor}, we obtain the result.
\end{proof}

By tensoring the right action~(\ref{action1}) this gives rise to
\begin{equation*}
\big(
\mathbb{C} \langle X \rangle \big(\!\big( w^{-1} \big)\!\big)
\otimes
\mathbb{C} \langle X \rangle \big(\!\big( u^{-1} \big)\!\big)
\big)
\times
\big(
\mathbb{C} \langle X \rangle [\![ z ]\!]
\otimes
\mathbb{C} \langle X \rangle [\![ z ]\!]
\big)
\to
\mathbb{C} \langle X \rangle \big(\!\big( w^{-1} \big)\!\big)
\otimes
\mathbb{C} \langle X \rangle \big(\!\big( u^{-1} \big)\!\big) ,
\end{equation*}
under which the image of
$( f \otimes g , x \otimes y )$
will be denoted by
$( f._{w} x ) \otimes ( g._{u} y )$
in the sequel.

We recall (see \cite[Proposition~3.3 and pp.~50--51]{SW})
that the tau-function
corresponding to $W \in \mathrm{Gr}$,
is given by
\begin{equation*}
\tau_{W} (\mathbf{t})=\det( w_{+} )=\det \big( 1 + a^{-1} b A \big),
\end{equation*}
up to a multiplicative constant,
where
$w_{+} \colon \mathrm{e}^{ \xi (\mathbf{t},z) }W \to H_{+}$,
is the orthogonal projection, and
$\mathrm{e}^{ \xi (\mathbf{t},z) } \colon H \to H$,
is the multiplication operator by
$\mathrm{e}^{ \xi (\mathbf{t},z) }$,
with matrix representation
\begin{equation*}
\mathrm{e}^{ - \xi (\mathbf{t},z) }
=
\left(\begin{matrix}
a & b \\
0 & d
\end{matrix}\right)
\qquad
\text{along $H = H_{+} \oplus H_{-}$},
\end{equation*}
and
$A \colon H_{+} \to H_{-}$
is such that
$\mathrm{graph}(A) = W$.
Given a
bounded
univalent function
$f\colon \mathbb{D} \to \mathbb{C}$,
with $f(0) = 0$,
we denote by
$A_{f}\colon H_{+} \to H_{-}$
the linear map such that
$\mathrm{graph} (A_{f}) = W_{f}$.

\begin{Thm}\label{Tau-by-Witt} Let
$\{ f_{t} \}_{0 \leqslant t \leqslant T}$
be a univalent solution to the Loewner--Kufarev equation
such that~$f_{t} (\mathbb{D})$ is bounded for every $t \in [ 0, T ]$.
Then for each $h \in H_{+}$ and $\vert z \vert > 1$,
we have
\begin{gather*}
( A_{f_{t}} h ) (z)=
\Bigg[
\int
\underset{\substack{w=0, \\ u=0}}{\operatorname{Res\ }}
\Bigg(
	\frac{ h^{\prime}(u) }{ w-z }
	\sum_{r,s=1}^{\infty}
	\mathrm{e}^{ (r+s) x_{0} (t) }
	x_{r+s}
	\big( w^{-r}._{w} S\big( \xi \big( \mathbf{x} , \mathrm{e}^{x_{0}(t)} \big) \big) \big)\\
\hphantom{( A_{f_{t}} h ) (z)=}{}
	\shuffle
	\big( u^{-s}._{u} S\big( \xi \big( \mathbf{x} , \mathrm{e}^{x_{0}(t)} \big) \big) \big)
\Bigg)
\Bigg]_{0,t} .
\end{gather*}
\end{Thm}

The proof can be found in
Appendix~\ref{app:tau-witt}. From this, we obtain
\begin{Cor}
For each $n,m \in \mathbb{N}$, the coefficient $b_{-n,-m}(t)$, is equal to
\begin{gather*}
\Bigg[
	\int
	\underset{\substack{z=0, \\ u=0}}{\operatorname{Res\ }}
	\Bigg\{
	\underset{w=0}{\operatorname{Res\ }}
	\frac{ z^{m-1} u^{n-1} }{ w-z }
	\sum_{r,s=1}^{\infty}
	\mathrm{e}^{ (r+s) x_{0} (t) } x_{r+s}
	\big( w^{-r}._{w} S\big( \xi \big( \mathbf{x}, \mathrm{e}^{x_{0}(t)} \big) \big) \big)\\
\qquad{}
	\shuffle
	\big( u^{-s}._{u} S\big( \xi ( \mathbf{x}, \mathrm{e}^{x_{0}(t)} \big) \big) \big)
	\Bigg\}\Bigg]_{0,t}.
\end{gather*}
\end{Cor}

\appendix
\section{Appendix}\label{Appdx}
\subsection{Proof of Theorem~\ref{explicit-Taylor}}\label{app:explicit-Taylor}

By applying variation of constants to~(\ref{c_n}),
we obtain the following recurence relation
\begin{equation*}
c_{n} (t)
=
\mathrm{e}^{n x_{0} (t)}
\int_{0}^{t} \mathrm{e}^{ -n x_{0} (s) } \mathrm{d} x_{n} (s)
+
\sum_{k=1}^{n-1}
( k+1 )
\mathrm{e}^{n x_{0} (t)}
\int_{0}^{t} \mathrm{e}^{-n x_{0} (s)}
c_{k} (s) \mathrm{d} x_{n-k} (s),
\end{equation*}
for $n \geqslant 2$.
Multiplying by $\mathrm{e}^{-n x_{0}(t)}$,
this transforms to
\begin{equation*}
\mathrm{e}^{-n x_{0}(t)}
c_{n} (t)
=
\int_{0}^{t}
\mathrm{e}^{-n x_{0}(t)}
\mathrm{d} x_{n} (s)
+
\sum_{k=1}^{n-1}
(k+1)
\int_{0}^{t}
\mathrm{e}^{-(n-k) x_{0}(t)}
\mathrm{d} x_{n-k} (s)
\big(
	\mathrm{e}^{-k x_{0}(s)}
	c_{k} (s)
\big) .
\end{equation*}
By assuming
that $x_{1}, x_{2}, \ldots$
are \textit{non-commutative} indeterminates,
and the $c_{n}$'s are polynomials in the $x_{i}$'s,
we shall consider the following equation:
\begin{equation}\label{Taylor_ev}
c_{n}
=
x_{n}
+ 2 x_{n-1} c_{1}
+ 3 x_{n-2} c_{2}
+ \cdots
+ (n-1) x_{2} c_{n-2}
+ n x_{1} c_{n-1},
\end{equation}
for $n \geqslant 1$
(roughly speaking,
the polynomial $c_{n}$
means
$\mathrm{e}^{-n x_{0}(t)} c_{n}(t)$
and
`applying the indeterminate $x_{k}$ from the left'
means `applying
$\int_{0}^{t} \mathrm{e}^{-kx_{0}(s)} \mathrm{d} x_{k}(s) \times $
to functions of~$s$')
and then we shall make some observations about the equation~(\ref{Taylor_ev})
and introduce some notations:
If we apply~(\ref{Taylor_ev}) to~$c_{n}$, we get
\begin{itemize}\itemsep=0pt
\item[(a)]
The terms
$(n-k+1) x_{k} c_{n-k}$
for each $k = 1, 2, \ldots , n$.
We shall denote these situation by
\begin{gather*}
c_{n}
\stackrel{\text{{\scriptsize $\widetilde{w}_{n,k} x_{k}$}}}{\to}
c_{n-k},
\end{gather*}
respectively (note that the multiplication by the $x_{*}$'s must sit
just {\it left} to the next $c_{*}$'s),
where
$\widetilde{w}_{n,k} := ((n-k)+1)$.

\item[(b)] The term $x_{0}$,
to which we can not apply~(\ref{Taylor_ev})
anymore.
This means,
consider the situation that
we apply~(\ref{Taylor_ev}) iteratively
to $c_{*}$'s which appeared at a previous stage.
Suppose we have the term~$c_{n}$ at some stage.
Then chasing the term multiplied by~$x_{*}$ which arose from the first term on the right-hand side
in~(\ref{Taylor_ev}),
lets us to get out of
the loop of iterations;
we shall symbolise this situation by
\begin{equation*}
c_{n}
\stackrel{\text{{\scriptsize $x_{n}$}}}{\rightrightarrows}
\text{end}.
\end{equation*}
\end{itemize}
Let $p \in \mathbb{N}$ be such that
$1 \leqslant p \leqslant n$.
We fix
$i_{1}, \ldots , i_{p} \in \mathbb{N}$,
so that
$i_{1} + \cdots + i_{p} = n$.
This data permits one to get out of
the loop of iterations of~(\ref{Taylor_ev})
as the following diagram shows:
\begin{gather*}\begin{split}
& \xymatrix@=20pt{
	c_{n}
	\ar[r]^-{ \text{{\scriptsize $\widetilde{w}_{n,i_{1}} x_{i_{1}}$}} }
	&
	c_{n-i_{1}}
	\ar[r]^-{ \text{{\scriptsize $\widetilde{w}_{n-i_{1},i_{2}} x_{i_{2}}$}} }
	&
	c_{n-i_{1}-i_{2}}
	\ar[rrr]^-{ \text{{\scriptsize $\widetilde{w}_{n-i_{1}-i_{2}, i_{3}} x_{i_{3}}$}} }
	&
	&
	&
	\cdots
	\ar[rrr]^-{ \text{{\scriptsize $\widetilde{w}_{n-(i_{1}+\cdots +i_{p-2}), i_{p-1}} x_{i_{p-1}}$}} }
	&
	&
	&
	c_{n-(i_{1}+i_{2}+\cdots + i_{p-1})}&}\\
& \xymatrix@=20pt{ &&&&&&&&&&= c_{i_{p}}
	\ar@<0.5ex>[r]^-{ \text{{\scriptsize $x_{i_{p}}$}} }
	\ar@<-0.5ex>[r]
	&
	\text{end}.
}\end{split}
\end{gather*}
Hence we have a single path
from $c_{n}$ to the `$\text{end}$'
in the above diagram.
This path produces at the `$\text{end}$' the term
\begin{equation*}
\widetilde{w}(n)_{i_{1},\ldots ,i_{p}}
x_{i_{p}} x_{i_{p-1}} \cdots x_{i_{2}} x_{i_{1}},
\end{equation*}
where, by using the relation
$
\widetilde{w}_{n-k,l} = \widetilde{w}_{n, k+l}
$,
the coefficient
$
\widetilde{w}(n)_{i_{1},\ldots ,i_{p}}
$
is given by
\begin{align*}
\widetilde{w}(n)_{i_{1},\ldots ,i_{p}}
& =
\widetilde{w}_{n,i_{1}}
\widetilde{w}_{n-i_{1}, i_{2}}
\widetilde{w}_{n-i_{1}-i_{2}, i_{3}}
\cdots
\widetilde{w}_{n-(i_{1}+i_{2}+ \cdots + i_{p-2}), i_{p-1}} \\
& =
\widetilde{w}_{n, i_{1}}
\widetilde{w}_{n, i_{1}+i_{2}}
\widetilde{w}_{n, i_{1}+i_{2}+i_{3}}
\cdots
\widetilde{w}_{n, i_{1}+i_{2}+ \cdots + i_{p-2} + i_{p-1}} \\
& =
\big\{ (n-i_{1}) + 1 \big\}
\big\{ (n-(i_{1}+i_{2})) + 1 \big\}
\cdots
\big\{ (n-(i_{1}+i_{2}+ \cdots +i_{p-1})) + 1 \big\} .
\end{align*}
Collecting all possibilities, we have
\begin{equation*}
c_{n}
=
\sum_{p=1}^{n}
\sum_{\substack{
	i_{1}, \ldots , i_{p} \in \mathbb{N}\colon \\
	i_{1} + \cdots + i_{p} = n
}}
\widetilde{w}(n)_{i_{1},\ldots ,i_{p}}
x_{i_{p}} x_{i_{p-1}} \cdots x_{i_{2}} x_{i_{1}},
\end{equation*}
which yields the result by reinterpreting it in the language of paths
$x_{k}(t)$'s, as claimed.

\subsection{Proof of Proposition~\ref{explicit-Grunsky}}\label{app:Grunsky-explicit}

By applying variation of constants to (\ref{Grunsky_evol}), we have
\begin{gather*}
b_{-m,-n} (t)=
-\mathrm{e}^{ (n+m) x_{0}(t) }
\int_{0}^{t}
\mathrm{e}^{ -(n+m) x_{0}(t) }
\mathrm{d} x_{n+m} (s) \\
\hphantom{b_{-m,-n} (t)=}{} +
\mathrm{e}^{ (n+m) x_{0}(t) }
\int_{0}^{t}
\big\{
	(n-1) b_{-m,-(n-1)} (s)
	\mathrm{d} x_{1} (s)
	+ \cdots
	+
	b_{-m,-1} (s)
	\mathrm{d} x_{n-1} (s)
\big\} \\
\hphantom{b_{-m,-n} (t)=}{}+
\mathrm{e}^{ (n+m) x_{0}(t) }
\int_{0}^{t}
\big\{
	(m-1) b_{-(m-1),-n} (s)
	\mathrm{d} x_{1} (s)
	+ \cdots
	+
	b_{-1,-n} (s)
	\mathrm{d} x_{m-1} (s)
\big\} .
\end{gather*}
By assuming
that $x_{1}, x_{2}, \ldots$
are {\it non-commutative} indeterminates, and the $b_{-m,-n}$'s polynomials in the $x_{i}$'s, we shall consider the following equation:
\begin{gather}
b_{-m,-n}
=
- x_{n+m}
+
\big\{
	(n-1) b_{-m,-(n-1)}
	x_{1}
	+ \cdots
	+ 2 b_{-m,-2} x_{n-2}
	+ b_{-m,-1} x_{n-1}
\big\} \nonumber\\
\hphantom{b_{-m,-n}=}{}+
\big\{
	(m-1) b_{-(m-1),-n} x_{1}
	+ \cdots
	+ 2 b_{-2,-n} x_{m-2}
	+ b_{-1,-n} x_{m-1}
\big\},\label{Grunsky_evol3}
\end{gather}
(roughly speaking,
the polynomial $b_{-m,-n}$
means
$
\mathrm{e}^{-(m+n)x_{0}(t)}b_{-m,-n}(t)
$
and
`applying the indeterminate $x_{k}$ from the right'
means `applying
$\int_{0}^{t} \mathrm{e}^{-kx_{0}(s)} \mathrm{d} x_{k}(s) \times $
to functions of~$s$'). If we apply~(\ref{Grunsky_evol3}) to $b_{-m,-n}$, we get:
\begin{itemize}\itemsep=0pt
\item[(a)]
The terms
$(n-k) b_{-m,-(n-k)} x_{k}$
and
$(m-k) b_{-(m-k),-n} x_{k}$
for each $k$.
We shall denote these cases by
\begin{equation*}
b_{-m,-n}
\stackrel{\text{{\tiny $(n- k ) x_{k} \times$}}}{\to}
b_{-m,-(n-k)}
\qquad
\text{and}
\qquad
\begin{matrix}
b_{-m,-n} \\
\text{{\tiny $(m- k ) x_{k} \times$}}
\downarrow \phantom{\text{{\tiny $(m-1) x_{k} \times$}}}
\\
b_{-(m-k),-n}
\end{matrix}\!\!\!\!\!\!\!\!\!\!,
\end{equation*}
respectively
(Note that the multiplication by the $x_{*}$'s must sit
just right to the next $b_{*,*}$'s).

\item[(b)]
The term $-x_{n+m}$,
to which we can not apply
(\ref{Grunsky_evol3})
anymore.
This means,
consider the situation that
we apply (\ref{Grunsky_evol3}) iteratively
to the $b_{*,*}$'s which appeared at a previous stage.
Suppose that we have the term $b_{-m,-n}$ at some stage.
Then chasing the term, multiplied by
$-x_{*}$, which arose from the first term on the right-hand side
in (\ref{Grunsky_evol3}),
permits us to get out of
the loop of iterations.
We shall denote this situation by
\begin{equation*}
b_{-m,-n}
\stackrel{\text{{\tiny $-x_{n+m} \times$}}}{\rightrightarrows}
\text{end}
\qquad
\text{or}
\qquad
\begin{matrix}
b_{-m,-n} \\
\text{{\tiny $-x_{n+m} \times$}}
\downdownarrows \phantom{\text{{\tiny $-x_{n+m} \times$}}}
\\
\text{end}
\end{matrix}\!\!\!\!\!\!\!\!\!\!\! .
\end{equation*}
Note that the multiplication by the $x_{*}$'s must be from the left.
Hence in particular,
to get the term of the form $x_{k}(\cdots)$
in the polynomial expression of $b_{-m,-n}$
in the~$x_{i}$'s,
we have to escape the loop by passing to the cases
\begin{equation*}
b_{-i,-j}
\stackrel{\text{{\tiny $-x_{k} \times$}}}{\rightrightarrows}
\text{end}
\qquad
\text{or}
\qquad
\begin{matrix}
b_{-i,-j} \\
\text{{\tiny $-x_{k} \times$}}
\downdownarrows
\phantom{\text{{\tiny $-x_{k} \times$}}} \\
\text{end}
\end{matrix}\!\!\!\!\!\!\!,
\end{equation*}
where $i,j \in \mathbb{N}$ with $i+j = k$.

\item[(c)]
If we have $b_{-1,-1}$, applying~(\ref{Grunsky_evol3})
does not produce $b_{*,*}$'s.
Namely we must have
\begin{equation*}
b_{-1,-1}
\stackrel{\text{{\tiny $- x_{2} \times$}}}{\rightrightarrows}
\text{end}
\qquad
\text{or}
\qquad
\begin{matrix}
b_{-1,-1} \\
\text{{\tiny $- x_{2} \times$}}
\downdownarrows
\phantom{\text{{\tiny $- x_{2} \times$}}} \\
\text{end}
\end{matrix}\!\!\!\!\!\!\! .
\end{equation*}
Again, the multiplication by $x_{2}$ must be from the left.
In particular, $b_{-m,-n}$ does not contain the term
$x_{1}(\cdots)$
and hence
$b_{-m,-n}$ is a linear combination of
$x_{k}(\cdots)$'s for $k\geqslant 2$,
though the factor $(\cdots)$ may involve $x_{1}$.
\end{itemize}
Let $k \in \mathbb{N}$ be such that
$2 \leqslant k \leqslant n+m$.
We shall find the term of the form $x_{k}(\cdots)$
in the polynomial expression of $b_{-m,-n}$
in the~$x_{i}$'s.
For this, we shall fix
$i \in \{ 1, \ldots , m \}$
and
$j \in \{ 1, \ldots , n \}$
such that
$i+j = k$.
Suppose that
$p,q \in \mathbb{N}$
and
$
i_{1}, \ldots , i_{p},
j_{1}, \ldots , j_{q} \in \mathbb{N}
$
satisfy
$
i_{1} + \cdots + i_{p} = m - i
$
and
$
j_{1} + \cdots + j_{q} = n - j
$.
We then put
$a_{r} := m - ( i_{1} + \cdots + i_{r}) $
for $r=1, \ldots , p$
and
$c_{s} := n - ( j_{1} + \cdots + j_{s} ) $
for $s=1, \ldots , q$.
Note that
$a_{p} = i$
and
$c_{q} = j$.
According to this notation,
we distinguish the following
three cases:

(1)
If there exist such $p$, $q$,
$(i_{1}, \ldots , i_{p})$
and
$(j_{1}, \ldots , j_{q})$,
then we can consider the following diagram:
\begin{equation*}
\xymatrix@=25pt{
	b_{-m,-n}
	\ar[d]_{ a_{1} x_{i_{1}} \times }
	\ar[r]^-{ c_{1} x_{j_{1}} \times }
	&
	b_{-m,-c_{1}}
	\ar[d]_{ a_{1} x_{i_{1}} \times }
	\ar[r]^-{ c_{2} x_{j_{2}} \times }
	&
	\cdots
	\ar[r]^-{ c_{q} x_{j_{q}} \times }
	&
	b_{-m,-c_{q}} = b_{-m,-j}
	\ar[d]_{ a_{1} x_{i_{1}} \times }
	& \\
	b_{-a_{1},-n}
	\ar[d]_{ a_{2} x_{i_{2}} \times }
	\ar[r]^-{ c_{1} x_{j_{1}} \times }
	&
	b_{-a_{1},-c_{1}}
	\ar[d]_{ a_{2} x_{i_{2}} \times }
	\ar[r]^-{ c_{2} x_{j_{2}} \times }
	&
	\cdots
	\ar[r]^-{ c_{q} x_{j_{q}} \times }
	&
	b_{-a_{1},-c_{q}} = b_{-a_{1},-j}
	\ar[d]_{ a_{2} x_{i_{2}} \times }
	& \\
	\vdots
	\ar[d]_{ a_{p} x_{i_{p}} \times }
	&
	\vdots
	\ar[d]_{ a_{p} x_{i_{p}} \times }
	&
	\mbox{}
	&
	\vdots
	\ar[d]_{ a_{p} x_{i_{p}} \times }
	& \\
	b_{-i,-n}
	\ar[r]^-{ c_{1} x_{j_{1}} \times }
	&
	b_{-i,-c_{1}}
	\ar[r]^-{ c_{2} x_{j_{2}} \times }
	&
	\cdots
	\ar[r]^-{ c_{q} x_{j_{q}} \times }
	&
	b_{-i,-c_{q}} = b_{-i,-j}
	\ar@<0.5ex>[rd]^{ (-1) x_{k} \times }
	\ar@<-0.5ex>[rd]
	& \\
	&
	&
	&
	&
	\text{end}
}
\end{equation*}
During the loop of iterations of (\ref{Grunsky_evol3}),
we have
$
\binom{p+q}{p}
=
\binom{p+q}{q}
$-paths
from $b_{-m,-n}$ to the `$\text{end}$'
in the above diagram,
each of which produces terms
\begin{equation*}
\text{$
- w_{i_{1},\ldots ,i_{p}; j_{1}, \ldots , j_{q}}
x_{k}(\cdots)
$'s},
\end{equation*}
where
\begin{gather*}
w_{i_{1},\ldots ,i_{p}; j_{1}, \ldots , j_{q}} =
a_{1} a_{2} \cdots a_{p}
b_{1} b_{2} \cdots b_{q} \\
\hphantom{w_{i_{1},\ldots ,i_{p}; j_{1}, \ldots , j_{q}}}{} =
( m - i_{1} )
( m - (i_{1}+i_{2}) )
\cdots
( m - (i_{1}+i_{2}+\cdots +i_{p}) ) \\
\hphantom{w_{i_{1},\ldots ,i_{p}; j_{1}, \ldots , j_{q}}=}{} \times
( n - j_{1} )
( n - (j_{1}+j_{2}) )
\cdots
( n - (j_{1}+j_{2}+\cdots +j_{q}) ),
\end{gather*}
(note that $w_{i_{1},\ldots ,i_{p}; j_{1}, \ldots , j_{q}}$
depends only on
$i_{1},\ldots ,i_{p}$ and $j_{1}, \ldots , j_{q}$
but not on the choice of paths in the diagram)
and $(\cdots)$ is a monomial consisting of
$x_{i_{p}}, x_{i_{p-1}}, \ldots , x_{i_{1}}$
and $x_{j_{q}}, x_{j_{q-1}}, \ldots , x_{j_{1}}$,
which is interlacing according to a riffle shuffle permutation
(note that we should distinguish,
for example
$x_{i_{1}}x_{j_{1}}$
and $x_{j_{1}}x_{i_{1}}$
even if $i_{1} = j_{1}$).
Hence, in total all paths produce
\begin{equation*}
- w_{i_{1}, \ldots , i_{p}; j_{1}, \ldots , j_{q}}
x_{k}
\big(
	( x_{i_{p}} x_{i_{p-1}} \cdots x_{i_{1}} )
	\shuffle
	( x_{j_{q}} x_{j_{q-1}} \cdots x_{j_{1}} )
\big) .
\end{equation*}

(2)
If there exist such a $p$ and
$(i_{1}, \ldots , i_{p})$
but not for
$q$ and $(j_{1}, \ldots , j_{q})$
(then we have $j=n$),
then the diagram which we can have is the following:
\begin{equation*}
\xymatrix@=20pt{
	b_{-m,-n}
	\ar[d]_{ a_{1} x_{i_{1}} \times } \\
	b_{-a_{1},-n}
	\ar[d]_{ a_{2} x_{i_{2}} \times } \\
	\vdots
	\ar[d]_{ a_{p} x_{i_{p}} \times } \\
	b_{-i,-n} = b_{-i,-j}
	\ar@<0.5ex>[d]
	\ar@<-0.5ex>[d]_{ (-1) x_{k} \times } \\
	\text{end}
}
\end{equation*}
Hence we have a single path
from $b_{-m,-n}$ to the `$\text{end}$'
in the above diagram.
This path produces the term
\begin{equation*}
- w_{i_{1},\ldots ,i_{p}}
x_{k}( x_{i_{p}} \cdots x_{i_{2}} x_{i_{1}} ),
\end{equation*}
where
$w_{i_{1},\ldots ,i_{p}; j_{1}, \ldots , j_{q}}
=
a_{1} a_{2} \cdots a_{p}
=
( m - i_{1} )
( m - (i_{1}+i_{2}) )
\cdots
( m - (i_{1}+i_{2}+\cdots +i_{p}) )
$.

(3) If there exist such a $q$ and
$(j_{1}, \ldots , j_{q})$
but not for
$p$ and $(i_{1}, \ldots , i_{p})$
(then we have $i=m$),
then the diagram which we can have is the following:
\begin{equation*}
\xymatrix@=20pt{
	b_{-m,-n}
	\ar[r]^-{ c_{1} x_{j_{1}} \times }
	&
	b_{-m,-c_{1}}
	\ar[r]^-{ c_{2} x_{j_{2}} \times }
	&
	\cdots
	\ar[r]^-{ c_{q} x_{j_{q}} \times }
	&
	b_{-m,-c_{q}} = b_{-i,-j}
	\ar@<0.5ex>[r]^-{ (-1) x_{k} \times }
	\ar@<-0.5ex>[r]
	&
	\text{end}.
}
\end{equation*}
Hence we have a single path
from $b_{-m,-n}$ to the `$\text{end}$'
in the above diagram.
This path produces the term
\begin{equation*}
- w_{j_{1},\ldots ,j_{q}}
x_{k}( x_{j_{q}} \cdots x_{j_{2}} x_{j_{1}} ),
\end{equation*}
where
\begin{gather*}
w_{j_{1}, \ldots , j_{q}}=c_{1} c_{2} \cdots c_{q} =
( n - j_{1} )
( n - (j_{1}+j_{2}) )
\cdots
( n - (j_{1}+j_{2}+\cdots +j_{q}) ) .
\end{gather*}
Now by reinterpreting it in the language of paths $x_{k}(t)$'s, we obtain the result.

\subsection{Proof of Theorem~\ref{Tau-by-Witt}}\label{app:tau-witt}

Since $\{ u^{n} \}_{n \geqslant 1}$ forms
a basis of $H_{+}$, it is enough to show that
\begin{gather}
n
\left[
\int
\underset{\substack{w=0, \\ u=0}}{\operatorname{Res\ }}
\left(
	\frac{ u^{n-1} }{ w-z }
	\sum_{r,s=1}^{\infty}
	\mathrm{e}^{ (r+s) x_{0} (t) }
	x_{r+s}
	\big( w^{-r}._{w} S\big( \xi \big( \mathbf{x} , \mathrm{e}^{x_{0}(t)} \big) \big) \big)
	\shuffle
	\big( u^{-s}._{u} S\big( \xi \big( \mathbf{x} , \mathrm{e}^{x_{0}(t)} \big) \big) \big)
\right)
\right]_{t} \nonumber\\
\qquad{} =
n
\sum_{m=1}^{\infty} b_{-n,-m} (t) z^{-m},\label{reduction1}
\end{gather}
where $b_{-n,-m}(t)$ are the Grunsky coefficients
associated with $f_{t}$.

According to the decomposition
\begin{equation*}
S ( \xi ( \mathbf{x} , \mathrm{e}^{x_{0}(t)} ) )
=
1
+
\sum_{ m^{\prime} = 1 }^{\infty}
\mathrm{e}^{ m^{\prime} x_{0}(t) }
\sum_{p=1}^{ m^{\prime} }
\sum_{
	\substack{
		i_{1}, \ldots , i_{p} \in \mathbb{N}: \\
		i_{1} + \cdots + i_{p} = m^{\prime}
	}
}
\!\!\!\!\!\!\!\!\!x_{i_{1}} x_{i_{2}} \cdots x_{i_{p}},
\end{equation*}
we have
\begin{gather*}
\big(
	w^{-r} ._{w} S\big( \xi \big( \mathbf{x} , \mathrm{e}^{x_{0}(t)} \big) \big)
\big)
\shuffle
\big(
	u^{-s} ._{u} S\big( \xi \big( \mathbf{x}, \mathrm{e}^{x_{0}(t)} \big) \big)
\big) \\
\qquad{} =
\big( w^{-r} ._{w} 1 \big)
\shuffle
\big( u^{-s} ._{u} 1 \big)
+
F_{r,s} (w,u)
+
G_{r,s} (w,u)
+
H_{r,s} (w,u),
\end{gather*}
where
\begin{gather*}
F_{r,s} (w,u)
:=
\big[
	w^{-r} ._{w}
	\big( S\big( \xi \big( \mathbf{x} , \mathrm{e}^{x_{0}(t)} \big) \big) - 1 \big)
\big]
\shuffle
\big[
	u^{-s} ._{u}
	\big( S\big( \xi \big( \mathbf{x}, \mathrm{e}^{x_{0}(t)} \big) \big) - 1 \big)
\big] \\
\hphantom{F_{r,s} (w,u)}{} =
\sum_{m^{\prime} =1}^{\infty}
\sum_{n^{\prime} =1}^{\infty}
\sum_{p=1}^{m^{\prime}}
\sum_{q=1}^{n^{\prime}}
\sum_{
	\substack{
		i_{1}, \ldots , i_{p} \in \mathbb{N}\colon \\
		i_{1} + \cdots + i_{p} = m^{\prime}
	}
}
\sum_{
	\substack{
		j_{1}, \ldots , j_{q} \in \mathbb{N}\colon \\
		j_{1} + \cdots + j_{q} = n^{\prime}
	}
}
\mathrm{e}^{ ( m^{\prime} + n^{\prime} ) x_{0}(t)} \\
\hphantom{F_{r,s} (w,u)=}{} \times
\big( w^{-r} ._{w} x_{i_{1}} x_{i_{2}} \cdots x_{i_{p}} \big)
\shuffle
\big( u^{-s} ._{u} x_{j_{1}} x_{j_{2}} \cdots x_{j_{q}} \big), \\
G_{r,s} (w,u):=
\big( w^{-r} ._{w} 1 \big)
\shuffle
\big[
	u^{-s} ._{u}
	\big( S\big( \xi \big( \mathbf{x}, \mathrm{e}^{x_{0}(t)} \big) \big) - 1 \big)
\big] \\
\hphantom{G_{r,s} (w,u)}{} =
\sum_{n^{\prime}=1}^{\infty}
\mathrm{e}^{ n^{\prime} x_{0} (t) }
\sum_{q=1}^{n^{\prime}}
\sum_{
	\substack{
		j_{1}, \ldots , j_{q} \in \mathbb{N}\colon \\
		j_{1} + \cdots + j_{q} = n^{\prime}
	}
}
\big( w^{-r} ._{w} 1 \big)
\shuffle
\big( u^{-s} ._{u} x_{j_{1}} x_{j_{2}} \cdots x_{j_{q}} \big), \\
H_{r,s} (w,u):=
\big[
	w^{-r} ._{w}
	\big( S\big( \xi \big( \mathbf{x}, \mathrm{e}^{ x_{0} (t) } \big) \big) - 1 \big)
\big]
\shuffle
\big( u^{-s} ._{u} 1 \big) \\
\hphantom{H_{r,s} (w,u)}{} =
\sum_{m^{\prime} =1}^{\infty}
\mathrm{e}^{ m^{\prime} x_{0} (t) }
\sum_{p=1}^{m^{\prime}}
\sum_{
	\substack{
		i_{1}, \ldots , i_{p} \in \mathbb{N}\colon \\
		i_{1} + \cdots + i_{p} = m^{\prime}
	}
}
\big( w^{-r} ._{w} x_{i_{1}} x_{i_{2}} \cdots x_{i_{p}} \big)
\shuffle
\big( u^{-s} ._{u} 1 \big) .
\end{gather*}
Since
$w^{-r} ._{w} 1 = w^{-r}$,
we get
$
\big( w^{-r} ._{w} 1 \big) \shuffle \big( u^{-s} ._{u} 1 \big)=w^{-r} u^{-s}$.
Then, by using
\begin{equation*}
\frac{ 1 }{ w-z }=-\sum_{m=1}^{\infty}z^{-m} w^{m-1}
\qquad
\text{for $\vert z \vert > \vert w \vert$,}
\end{equation*}
we have
\begin{gather*}
\underset{\substack{w=0; \\ u=0}}{\operatorname{Res\ }}
\left(
\frac{ u^{n-1} }{ w-z }
\sum_{r,s=1}^{\infty}
\mathrm{e}^{ (r+s) x_{0} (t) }
x_{r+s}
\big(
	\big( w^{-r} ._{w} 1 \big)
	\shuffle
	\big( u^{-s} ._{u} 1 \big)
\big)
\right) \\
\qquad{} = -
\underset{\substack{w=0; \\ u=0}}{\operatorname{Res\ }}
\left(
u^{n-1}
\sum_{m=1}^{\infty}
z^{-m} w^{m-1}
\sum_{r,s=1}^{\infty}
\mathrm{e}^{ (r+s) x_{0} (t) }
x_{r+s}
w^{-r} u^{-s}
\right)\\
\qquad{} =
-
\sum_{m=1}^{\infty}
z^{-m}
\mathrm{e}^{ (m+n) x_{0} (t) }
x_{m+n} .
\end{gather*}
Let
\begin{gather*}
w(r)_{i_{1},\ldots , i_{p}; \varnothing}
:=
r
( i_{p} + r )
( i_{p} + i_{p-1} + r )
\cdots
( i_{p} + i_{p-1} + \cdots + i_{2} + r ) \\
\qquad{} =
\big( m - ( i_{1} + \cdots + i_{p} ) \big)
\cdots
\big( m - ( i_{1} + i_{2} ) \big)
( m - i_{1} ),
\end{gather*}
where $m = i_{1} + \cdots + i_{p} + r$,
\begin{gather*}
w(s)_{\varnothing ; j_{1},\ldots , j_{q}}
 :=
s
( j_{q} + s )
( j_{q} + j_{q-1} + s )
\cdots
( j_{q} + j_{q-1} + \cdots + j_{2} + s ) \\
\qquad{} =
\big( n - ( j_{1} + \cdots + j_{q} ) \big)
\cdots
\big( n - ( j_{1} + j_{2} ) \big)
( n - j_{1} ),
\end{gather*}
where $n = j_{1} + \cdots + j_{q} + s$,
and
\[
w(r,s)_{ i_{1}, \ldots , i_{p} ; j_{1}, \ldots , j_{q} }
:=
w(r)_{i_{1},\ldots , i_{p}; \varnothing}
w(s)_{\varnothing ; j_{1},\ldots , j_{q}}.
\]

For $F_{r+s} (w,u)$, we first observe that
\begin{gather*}
w^{-r}._{w} x_{i_{p}} \cdots x_{i_{2}} x_{i_{1}}
=x_{i_{p}} \cdots x_{i_{2}} x_{i_{1}}
L_{-i_{1}}
L_{-i_{2}}
\cdots
L_{-i_{p}}
w^{-r} \\
\hphantom{w^{-r}._{w} x_{i_{p}} \cdots x_{i_{2}} x_{i_{1}} }{} =
w(r)_{ i_{1}, \ldots , i_{p}; \varnothing }
x_{i_{p}} \cdots x_{i_{2}} x_{i_{1}}
w^{ - ( i_{1} + i_{2} + \cdots + i_{p} + r ) },
\end{gather*}
and similarly
\begin{gather*}
u^{-s}._{u} x_{j_{q}} \cdots x_{j_{2}} x_{j_{1}}
=
w(s)_{ \varnothing ; j_{1}, \ldots , j_{q} }
x_{j_{q}} \cdots x_{j_{2}} x_{j_{1}}
u^{ - ( j_{1} + j_{2} + \cdots + j_{p} + s ) }.
\end{gather*}
Therefore we have
\begin{gather*}
\underset{\substack{w=0 ; \\ u=0}}{\operatorname{Res\ }}
\left(
\frac{ u^{n-1} }{ w - z }
x_{r+s}
\big(
	( w^{-r}._{w} x_{i_{1}} x_{i_{2}} \cdots x_{i_{p}} )
	\shuffle
	( u^{-s}._{u} x_{i_{1}} x_{i_{2}} \cdots x_{i_{p}} )
\big)
\right) \\
\qquad{} =
-
1_{ \{ 1 \leqslant n-s = j_{1} + \cdots + j_{q} \} }
\sum_{m=1}^{\infty}
z^{-m}
1_{ \{ 1 \leqslant m-r = i_{1} + \cdots + i_{p} \} } \\
\qquad\quad{} \times
w(r,s)_{ i_{1}, \ldots , i_{p} ; j_{1}, \ldots , j_{q} }
x_{r+s}
[
	( x_{i_{p}} \cdots x_{i_{2}} x_{i_{1}} )
	\shuffle
	( x_{j_{q}} \cdots x_{j_{2}} x_{j_{1}} )
] ,
\end{gather*}
so that
\begin{gather*}
\underset{\substack{w=0 ; \\ u=0}}{\operatorname{Res\ }}
\left(
\frac{ u^{n-1} }{ w - z }
x_{r+s}
F_{r,s} (w,u)
\right) \\
{} =
-\sum_{m=1}^{\infty}
z^{-m}
\sum_{m^{\prime} =1}^{\infty}
\sum_{n^{\prime} =1}^{\infty}
\mathrm{e}^{ (m^{\prime} + n^{\prime}) x_{0} (t) }
\sum_{p=1}^{m^{\prime}}
\sum_{q=1}^{n^{\prime}}
\sum_{
	\substack{
		i_{1}, \ldots , i_{p} \in \mathbb{N}: \\
		i_{1} + \cdots + i_{p} = m^{\prime}
	}
}
\!\!\!\!\!\!\!1_{ \{ 1 \leqslant m-r = i_{1} + \cdots + i_{p} \} }\cdots\\
\quad{} \cdots\sum_{
	\substack{
		j_{1}, \ldots , j_{q} \in \mathbb{N}\colon \\
		j_{1} + \cdots + j_{q} = n^{\prime}
	}
}
\!\!\!1_{ \{ 1 \leqslant n-s = j_{1} + \cdots + j_{q} \} }
w(r,s)_{ i_{1}, \ldots , i_{p} ; j_{1}, \ldots , j_{q} }
x_{r+s}
[
	( x_{i_{p}} \cdots x_{i_{2}} x_{i_{1}} )
	\shuffle
	( x_{j_{q}} \cdots x_{j_{2}} x_{j_{1}} )
] \\
{} =
-1_{ \{ 1 \leqslant n-s \} }
\sum_{m=1}^{\infty}
z^{-m}
\mathrm{e}^{ ( (m-r) + (n-s) ) x_{0} (t) }
1_{ \{ 1 \leqslant m-r \} }
\sum_{p=1}^{m-r}
\sum_{q=1}^{n-s}
\sum_{
	\substack{
		i_{1}, \ldots , i_{p} \in \mathbb{N}: \\
		i_{1} + \cdots + i_{p} = m-r
	}
}\cdots\\
\quad{} \cdots\sum_{
	\substack{
		j_{1}, \ldots , j_{q} \in \mathbb{N}:\\
		j_{1} + \cdots + j_{q} = n-s
	}
}
\!\!\!\!\!\!\!\!\!w(r,s)_{ i_{1}, \ldots , i_{p} ; j_{1}, \ldots , j_{q} }
x_{r+s}
[
	( x_{i_{p}} \cdots x_{i_{2}} x_{i_{1}} )
	\shuffle
	( x_{j_{q}} \cdots x_{j_{2}} x_{j_{1}} )
] .
\end{gather*}
Hence we have reached
\begin{gather*}
\underset{\substack{w=0 ; \\ u=0}}{\operatorname{Res\ }}
\left(
\frac{ u^{n-1} }{ w - z }
\sum_{r,s =1}^{\infty}
\mathrm{e}^{ (r+s) x_{0} (t) }
x_{r+s}
F_{r,s} (w,u)
\right) \\
\qquad{} =
-
\sum_{m=1}^{\infty}
z^{-m}
\mathrm{e}^{ (m+n) x_{0} (t) }
\sum_{k=2}^{m+n-2}
\sum_{\substack{ 1 \leqslant r < m; \\ 1 \leqslant s < n\colon \\ r+s = k}}
\sum_{p=1}^{m-r}
\sum_{q=1}^{n-s}
\sum_{
	\substack{
		i_{1}, \ldots , i_{p} \in \mathbb{N}\colon \\
		i_{1} + \cdots + i_{p} = m-r
	}
}\cdots\\
\qquad\quad{} \cdots
\sum_{
	\substack{
		j_{1}, \ldots , j_{q} \in \mathbb{N}: \\
		j_{1} + \cdots + j_{q} = n-s
	}
}
\!\!\!\!\!\!\!\!\!w(r,s)_{ i_{1}, \ldots , i_{p} ; j_{1}, \ldots , j_{q} }
x_{k}
[
	( x_{i_{p}} \cdots x_{i_{2}} x_{i_{1}} )
	\shuffle
	( x_{j_{q}} \cdots x_{j_{2}} x_{j_{1}} )
] .
\end{gather*}
Similarly, we find that
\begin{gather*}
\underset{\substack{w=0 ; \\ u=0}}{\operatorname{Res\ }}
\left(
\frac{ u^{n-1} }{ w - z }
\sum_{r,s =1}^{\infty}
\mathrm{e}^{ (r+s) x_{0} (t) }
x_{r+s}
G_{r,s} (w,u)
\right) \\
\qquad{} = -
\sum_{m=1}^{\infty}
z^{-m}
\sum_{r,s=1}^{\infty}
\mathrm{e}^{ (r+s) x_{0} (t) }
\sum_{n^{\prime}=1}^{\infty}
\mathrm{e}^{ n^{\prime} x_{0} (t) }
\sum_{q=1}^{n^{\prime}}\cdots\\
\qquad\quad{} \cdots\sum_{
	\substack{
		j_{1}, \ldots , j_{q} \in \mathbb{N}\colon \\
		j_{1} + \cdots + j_{q} = n^{\prime}
	}
}
1_{\{ 1 \leqslant n-s = j_{1} + \cdots + j_{q} \}}
1_{\{ m=r \}}
w(s)_{ \varnothing ; j_{1}, \ldots , j_{q} }
x_{ r+s } ( x_{j_{q}} \cdots x_{j_{1}} ) \\
\qquad{} =
-\sum_{m=1}^{\infty}
z^{-m}
\mathrm{e}^{ (m+n) x_{0} (t) }
\sum_{k=m+1}^{m+n-1}
\sum_{q=1}^{n+m-k}
\!\!\!\sum_{
	\substack{
		j_{1}, \ldots , j_{q} \in \mathbb{N}: \\
		j_{1} + \cdots + j_{q} = n+m-k
	}
}
\!\!\!\!\!\!\!\!\!\!\!\!w( k-m )_{ \varnothing ; j_{1}, \ldots , j_{q} }
x_{ k } ( x_{j_{q}} \cdots x_{j_{1}} ),
\end{gather*}
and
\begin{gather*}
\underset{\substack{w=0 ; \\ u=0}}{\operatorname{Res\ }}
\left(
\frac{ u^{n-1} }{ w - z }
\sum_{r,s =1}^{\infty}
\mathrm{e}^{ (r+s) x_{0} (t) }
x_{r+s}
H_{r,s} (w,u)
\right) \\
\qquad{} =
-\sum_{m=1}^{\infty}
z^{-m}
\mathrm{e}^{ (m+n) x_{0} (t) }
\sum_{k=n+1}^{m+n-1}
\sum_{p=1}^{n+m-k}
\!\!\!\sum_{
	\substack{
		i_{1}, \ldots , i_{p} \in \mathbb{N}: \\
		i_{1} + \cdots + i_{p} = m+n-k
	}
}
\!\!\!\!\!\!\!\!\!w(k-n)_{ i_{1}, \ldots , i_{p} ; \varnothing }
x_{ k } ( x_{i_{p}} \cdots x_{i_{1}} ) .
\end{gather*}
Now, in view of Theorem \ref{Grunsky-formula}, we obtain (\ref{reduction1}), and hence the result.

\subsection*{Acknowledgements}

T.A.\ was supported by JSPS KAKENHI Grant Number 15K17562.
R.F.\ was previously supported by the ERC advanced grant ``Noncommutative distributions in free probability''.
Both authors thank Theo Sturm for the hospitality he offered to T.A.\ at the University of Bonn. T.A.\ thanks Roland Speicher for the hospitality offered him in Saarbr\"ucken.
R.F.\ thanks Roland Spei\-cher for discussions, and Fukuoka University and the MPI in Bonn for their hospitality.
We both thank Takuya Murayama for the discussions, comments and collaboration.
We thank the anonymous referees for their comments which helped us to improve the paper.

\pdfbookmark[1]{References}{ref}
\LastPageEnding


\begin{thebibliography}{99}
\footnotesize\itemsep=0pt

\bibitem{AGMV}
Alvarez-Gaum\'e L., Gomez C., Moore G., Vafa C., Strings in the operator
 formalism, \href{https://doi.org/10.1016/0550-3213(88)90391-4}{\textit{Nuclear Phys.~B}} \textbf{303} (1988), 455--521.

\bibitem{AFb}
Amaba T., Friedrich R., Modulus of continuity of controlled
 {L}oewner--{K}ufarev equations and random mat\-rices, \href{https://doi.org/10.1007/s13324-020-00366-3}{\textit{Anal. Math.
 Phys.}} \textbf{10} (2020), 23, 29~pages, \href{https://arxiv.org/abs/1809.00536}{arXiv:1809.00536}.

\bibitem{AFM}
Amaba T., Friedrich R., Murayama T., Univalence and holomorphic extension of
 the solution to {$\omega$}-controlled {L}oewner--{K}ufarev equations,
 \href{https://doi.org/10.1016/j.jde.2020.02.011}{\textit{J.~Differential Equations}} \textbf{269} (2020), 2697--2704,
 \href{https://arxiv.org/abs/1909.13666}{arXiv:1909.13666}.

\bibitem{Ba}
Baudoin F., An introduction to the geometry of stochastic flows, \href{https://doi.org/10.1142/9781860947261}{Imperial
 College Press}, London, 2004.

\bibitem{BCDV}
Bracci F., Contreras M.D., D\'{\i}az-Madrigal S., Vasil'ev A., Classical and
 stochastic {L}\"owner--{K}ufarev equations, in Harmonic and Complex Analysis
 and its Applications, \textit{Trends Math.}, \href{https://doi.org/10.1007/978-3-319-01806-5_2}{Birkh\"auser/Springer}, Cham, 2014,
 39--134.

\bibitem{Do1}
Doyon B., Conformal loop ensembles and the stress-energy tensor, \href{https://doi.org/10.1007/s11005-012-0594-1}{\textit{Lett.
 Math. Phys.}} \textbf{103} (2013), 233--284, \href{https://arxiv.org/abs/1209.1560}{arXiv:1209.1560}.

\bibitem{DNNZ}
Duplantier B., Nguyen C., Nguyen N., Zinsmeister M., The coefficient problem
 and multifractality of whole-plane {SLE} \& {LLE}, \href{https://doi.org/10.1007/s00023-014-0351-3}{\textit{Ann. Henri
 Poincar\'e}} \textbf{16} (2015), 1311--1395, \href{https://arxiv.org/abs/#2}{arXiv:1211.2451}.

\bibitem{El}
Ellacott S.W., A survey of {F}aber methods in numerical approximation,
 \href{https://doi.org/10.1016/0898-1221(86)90234-8}{\textit{Comput. Math. Appl. Part~B}} \textbf{12} (1986), 1103--1107.

\bibitem{Fa}
Faber G., \"Uber polynomische {E}ntwickelungen, \href{https://doi.org/10.1007/BF01444293}{\textit{Math. Ann.}} \textbf{57}
 (1903), 389--408.

\bibitem{Fr10}
Friedrich R., The global geometry of stochastic {L}{\oe}wner evolutions, in
 Probabilistic Approach to Geometry, \textit{Adv. Stud. Pure Math.}, Vol.~57,
 \href{https://doi.org/10.2969/aspm/05710079}{Math. Soc. Japan}, Tokyo, 2010, 79--117, \href{https://arxiv.org/abs/0906.5328}{arXiv:0906.5328}.

\bibitem{FrKa}
Friedrich R., Kalkkinen J., On conformal field theory and stochastic {L}oewner
 evolution, \href{https://doi.org/10.1016/j.nuclphysb.2004.03.025}{\textit{Nuclear Phys.~B}} \textbf{687} (2004), 279--302,
 \href{https://arxiv.org/abs/hep-th/0308020}{arXiv:hep-th/0308020}.

\bibitem{HiMaVa}
Hidalgo R.A., Markina I., Vasil'ev A., Finite dimensional grading of the
 {V}irasoro algebra, \href{https://doi.org/10.1515/GMJ.2007.419}{\textit{Georgian Math.~J.}} \textbf{14} (2007), 419--434.

\bibitem{Jo}
Johnston E., The {F}aber transform and analytic continuation, \href{https://doi.org/10.2307/2047558}{\textit{Proc.
 Amer. Math. Soc.}} \textbf{103} (1988), 237--243.

\bibitem{KNTY}
Kawamoto N., Namikawa Y., Tsuchiya A., Yamada Y., Geometric realization of
 conformal field theory on {R}iemann surfaces, \href{https://doi.org/10.1007/BF01225258}{\textit{Comm. Math. Phys.}}
 \textbf{116} (1988), 247--308.

\bibitem{KiYu}
Kirillov A.A., Yuriev D.V., Representations of the {V}irasoro algebra by the
 orbit method, \href{https://doi.org/10.1016/0393-0440(88)90029-0}{\textit{J.~Geom. Phys.}} \textbf{5} (1988), 351--363.

\bibitem{Ko}
Kontsevich M., CFT, SLE and phase boundaries, {P}reprint, Arbeitstagung, MPI
 Bonn, 2003.

\bibitem{Kr1}
Krichever I.M., Algebraic-geometric construction of the Zakharov--Shabat
 equations and their periodic solutions, \textit{Sov. Math. Dokl.} \textbf{17}
 (1976), 394--397.

\bibitem{Kr2}
Krichever I.M., Integration of nonlinear equations by the methods of algebraic
 geometry, \href{https://doi.org/10.1007/BF01135528}{\textit{Funct. Anal. Appl.}} \textbf{11} (1977), 12--26.

\bibitem{Kr3}
Krichever I.M., Methods of algebraic geometry in the theory of non-linear
 equations, \href{https://doi.org/10.1070/RM1977v032n06ABEH003862}{\textit{Russian Math. Surveys}} \textbf{32} (1977), no.~6,
 185--213.

\bibitem{Ku}
Kufareff P.P., On one-parameter families of analytic functions, \textit{Math.
 Sb.} \textbf{13(55)} (1943), 87--118.

\bibitem{Lo}
L\"owner K., Untersuchungen \"uber schlichte konforme {A}bbildungen des
 {E}inheitskreises.~{I}, \href{https://doi.org/10.1007/BF01448091}{\textit{Math. Ann.}} \textbf{89} (1923), 103--121.

\bibitem{LyCaLh}
Lyons T.J., Caruana M., L\'evy T., Differential equations driven by rough
 paths, \textit{Lecture Notes in Math.}, Vol.~1908, \href{https://doi.org/10.1007/978-3-540-71285-5}{Springer}, Berlin, 2007.

\bibitem{Ma99}
Malliavin P., The canonic diffusion above the diffeomorphism group of the
 circle, \href{https://doi.org/10.1016/S0764-4442(00)88575-4}{\textit{C.~R.~Acad. Sci. Paris S\'er.~I Math.}} \textbf{329} (1999),
 325--329.

\bibitem{MaDmVa}
Markina I., Prokhorov D., Vasil'ev A., Sub-{R}iemannian geometry of the
 coefficients of univalent functions, \href{https://doi.org/10.1016/j.jfa.2006.09.013}{\textit{J.~Funct. Anal.}} \textbf{245}
 (2007), 475--492, \href{https://arxiv.org/abs/math.CV/0608532}{arXiv:math.CV/0608532}.

\bibitem{MaVa10}
Markina I., Vasil'ev A., Virasoro algebra and dynamics in the space of
 univalent functions, in Five Lectures in Complex Analysis, \textit{Contemp.
 Math.}, Vol.~525, \href{https://doi.org/10.1090/conm/525/10365}{Amer. Math. Soc.}, Providence, RI, 2010, 85--116.

\bibitem{MaVa11}
Markina I., Vasil'ev A., L\"owner--Kufarev evolution in the Segal--Wilson
 Grassmannian, in Geometric Methods in Physics, \textit{Trends in
 Mathematics}, Editors P.~Kielanowski, S.T.~Ali, A.~Odzijewicz,
 M.~Schlichenmaier, T.~Voronov, \href{https://doi.org/10.1007/978-3-0348-0448-6_33}{Birkh\"auser/Springer}, Basel, 2013, 367--376.

\bibitem{MaVa16}
Markina I., Vasil'ev A., Evolution of smooth shapes and integrable systems,
 \href{https://doi.org/10.1007/s40315-015-0133-z}{\textit{Comput. Methods Funct. Theory}} \textbf{16} (2016), 203--229,
 \href{https://arxiv.org/abs/1108.1007}{arXiv:1108.1007}.

\bibitem{Mu}
Mumford D., An algebro-geometric construction of commuting operators and of
 solutions to the {T}oda lattice equation, {K}orteweg de{V}ries equation and
 related nonlinear equation, in Proceedings of the {I}nternational {S}ymposium
 on {A}lgebraic {G}eometry ({K}yoto {U}niv., {K}yoto, 1977), Editor M.~Nagata,
 Kinokuniya Book Store, Tokyo, 1978, 115--153.

\bibitem{Po2}
Pommerenke C., Univalent functions (with a chapter on quadratic differentials
 by Gerd Jensen), \textit{Studia Mathematica/Mathematische Lehrb\"ucher},
 Vol.~25, Vandenhoeck \& Ruprecht, G\"ottingen, 1975.

\bibitem{Sc}
Schiffer M., Faber polynomials in the theory of univalent functions,
 \href{https://doi.org/10.1090/S0002-9904-1948-09027-9}{\textit{Bull. Amer. Math. Soc.}} \textbf{54} (1948), 503--517.

\bibitem{Sch}
Schramm O., Scaling limits of loop-erased random walks and uniform spanning
 trees, \href{https://doi.org/10.1007/BF02803524}{\textit{Israel~J. Math.}} \textbf{118} (2000), 221--288,
 \href{https://arxiv.org/abs/math.PR/9904022}{arXiv:math.PR/9904022}.

\bibitem{SW}
Segal G., Wilson G., Loop groups and equations of {K}d{V} type, \href{https://doi.org/10.1007/BF02698802}{\textit{Inst.
 Hautes \'Etudes Sci. Publ. Math.}} \textbf{61} (1985), 5--65.

\bibitem{So}
Sola A., Elementary examples of {L}oewner chains generated by densities,
 \href{https://doi.org/10.2478/v10062-012-0024-y}{\textit{Ann. Univ. Mariae Curie-Sk{\l}odowska Sect.~A}} \textbf{67} (2013),
 83--101.

\bibitem{TT91}
Takasaki K., Takebe T., {${\rm SDiff}(2)$} {T}oda equation~-- hierarchy, tau
 function, and symmetries, \href{https://doi.org/10.1007/BF01885498}{\textit{Lett. Math. Phys.}} \textbf{23} (1991),
 205--214, \href{https://arxiv.org/abs/hep-th/9112042}{arXiv:hep-th/9112042}.

\bibitem{Teo03}
Teo L.-P., Analytic functions and integrable hierarchies~-- characterization of
 tau functions, \href{https://doi.org/10.1023/A:1024969729259}{\textit{Lett. Math. Phys.}} \textbf{64} (2003), 75--92,
 \href{https://arxiv.org/abs/hep-th/0305005}{arXiv:hep-th/0305005}.

\bibitem{WZ00}
Wiegmann P.B., Zabrodin A., Conformal maps and integrable hierarchies,
 \href{https://doi.org/10.1007/s002200000249}{\textit{Comm. Math. Phys.}} \textbf{213} (2000), 523--538,
 \href{https://arxiv.org/abs/hep-th/9909147}{arXiv:hep-th/9909147}.

\bibitem{Ya80}
Yamada A., Precise variational formulas for abelian differentials,
 \textit{Kodai Math.~J.} \textbf{3} (1980), 114--143.

\end{thebibliography}
\end{document}